\documentclass[english,thm-restate]{lipics-v2021}

\nolinenumbers
\hideLIPIcs 

\usepackage{enumitem}

\usepackage{cancel}
\usepackage{ulem}

\usepackage{tikz} 
\usetikzlibrary{patterns}
\usetikzlibrary{shapes.geometric, calc}

\usepackage{tcolorbox}
\usepackage{hyperref}
\usepackage{color}
\definecolor{darkgreen}{RGB}{0,100,0}

\definecolor{firebrick}{RGB}{178,34,34}
\definecolor{Andre_blue}{RGB}{0,120,220}
\definecolor{Andre_gray}{RGB}{180,200,200}

\theoremstyle{plain}

\DeclareMathOperator{\ax}{{ax}}

\newcommand{\ud}{\mathrm{d}}

\newcommand\R{\mathbb{R}}
\newcommand{\M}{\mathcal{M}}

\newcommand{\Su}{{\mathcal{S}}}

\newcommand{\lfs}{\textrm{lfs}}
\newcommand{\Tan}{\operatorname{Tan}}
\newcommand{\Nor}{\operatorname{Nor}}
\newcommand{\rch}{\mathrm{rch}}

\newcommand\myatop[2]{\genfrac{}{}{0pt}{}{#1}{#2}}

\newcommand{\defunder}[1]{\underset{\text{def.}}{#1} \:}
\newcommand{\Sphere}{\mathbb{S}}
\newcommand{\Grass}{\operatorname{Gr}}

\title{Manifolds of positive reach, differentiability, tangent variation, and attaining the reach
}

\author{Andr{\'e} Lieutier}{No affiliation\\{[Aix-en-Provence, France]}}{andre.lieutier@gmail.com }{}{}

\author{Mathijs Wintraecken}{Inria centre Universit{\'e} C{\^o}te d'Azur\\{[Sophia Antipolis, France]}}{mathijs.wintraecken@inria.fr}{https://orcid.org/0000-0002-7472-2220}{Supported by the European Union's Horizon 2020 research and innovation programme under the Marie Sk{\l}odowska-Curie grant agreement No. 754411. The Austrian science fund (FWF) M-3073. The welcome package from IDEX of the Universit{\'e} C{\^o}te d'Azur.  The French National Research Agency (ANR) under grant StratMesh. 
}

\authorrunning{Andr{\'e} Lieutier 
and Mathijs Wintraecken} 

\Copyright{ 
Andr{\'e} Lieutier,
 and Mathijs Wintraecken} 

\acknowledgements{We thank Jean-Daniel Boissonnat for discussion.  We would also like to acknowledge the organizers of the workshop on `Algorithms for the Medial Axis', and Erin Chambers in particular for giving an impulse to this research. 
}

\ccsdesc{Theory of computation $\rightarrow$ Computational geometry}

\keywords{Reach, Manifolds, Differentiability class, Lipschitz continuity, Tangent space}

\EventEditors{??? and ???}
\EventNoEds{2}
\EventLongTitle{41st International Symposium on Computational Geometry (SoCG 2025)}
\EventShortTitle{SoCG 2025}
\EventAcronym{SoCG}
\EventYear{2025}
\EventDate{June ???, 2025}
\EventLocation{Kanazawa, Japan}
\EventLogo{socg-logo.pdf}
\SeriesVolume{????}

\begin{document}

\maketitle

\begin{abstract}
This paper contains three main results. 

Firstly, we give an elementary proof of the following statement: Let $\M$ be a {(closed, in both the geometrical and topological sense of the word)} topological manifold embedded in $\R^d$. If $\M$ has positive reach, then $\M$ can locally be written as the graph of a $C^{1,1}$ from the tangent space to the normal space. Conversely if $\M$  can locally written as the graph of a $C^{1,1}$ function from the tangent space to the normal space, then $\M$ has positive reach.  
The result was hinted at by Federer when he introduced the reach, and proved by Lytchak. Lytchak's proof relies heavily CAT(k)-theory. The proof presented here uses only basic results on homology.

Secondly, we give optimal Lipschitz-constants for the derivative, in other words we give an optimal bound for the angle between tangent spaces in term of the distance between the points. This improves earlier results, that were either suboptimal or assumed that the manifold was $C^2$. 

Thirdly, we generalize a result by Aamari et al. which explains the how the reach is attained from the smooth setting to general sets of positive reach. 
\end{abstract}

\section{Introduction}


In \cite{Federer} Federer introduced the reach of a closed set $\Su \subset \mathbb{R}^d$ as the minimum of the distance from $\Su$ to the medial axis $\ax(\Su)$, i.e. the set of points in $\mathbb{R}^d$ for which the closest point in $\Su$ is not unique. Assumptions on the reach (and its local version the local feature size \cite{Amenta1999}) underpin the correctness of many algorithms in computational geometry and topology. 
In this paper we consider the differentiability class of manifolds of positive reach in particular of manifolds and give tight bounds on the angle between nearby tangent spaces for general manifolds of positive reach. For arbitrary sets of positive reach we give a geometrical explanation on how the reach in attained.


\paragraph*{Previous work} 
\subparagraph{Differentiability} 
Federer proved that the reach is stable under $C^{1,1}$-diffeomorphisms of the ambient space. Here we write $C^{1,1}$ to indicate $C^1$ maps whose derivative is Lipschitz, and by a $C^{1,1}$-diffeomorphism we mean that the both the diffeomorphism and its inverse are $C^{1,1}$.
Federer \cite[Remark 4.20]{Federer} furthermore mentioned (without going into much detail on one direction of the implication) that the graph of a function has positive reach if and only if it is $C^{1,1}$. Lytchak \cite{lytchak2004geometry, lytchak2005almost}
proved that a topologically embedded manifold without boundary has positive reach if and only if it is $C^{1,1}$, without a quantitative bound on the Lipschitz constant. This statement is quantified in our Theorem \ref{PosReachImpliesOptimalSizeNeighbourAsGraph} 
below.

Lytchak's proof requires significant background in CAT(k)-theory. 
Lytchak, motivated by bounds on intrinsic curvature, does not give any quantified bound on the extrinsic curvature, which play an important role in this paper.

Scholtes \cite{Scholtes2013} also gave a proof that a hypersurface has positive reach if and only if it is $C^{1,1}$, using different techniques from both the ones employed by Lytchak and in this paper. Scholtes' method uses the fact that the manifold has codimension one in an essential way.

Rataj and Zaj{\'i}\v{c}ek \cite{rataj2017structure} prove that if a Lipschitz manifold has positive reach, then it is $C^{1,1}$. They also use Lemma \ref{lemma:SemiConcaveAndSemiConvexIffC11}, however the assumption that the manifold is already Lipschitz simplifies the matter considerably, in particular they can skip the topological analysis. 
Moreover, they do not quantify their results in the way that we do and is essential for applications.

\subparagraph{Geometric bounds and tangent variation} Bounding the tangent angle variation on smooth manifolds is crucial for establishing constants in the context manifolds reconstruction and meshing \cite{Amenta1999, ChengDeyShewchuck, Dey, JDbook}. 
If we now focus on more recent results, some authors consider angle variation bounds for a given local feature size bound \cite{khoury2019approximation}.
Other authors 
give angle variation bounds based on the reach \cite{TangentVarJournal}.  
Assuming the surface to be smooth, more specifically $C^2$ smooth, 
allows  using standard differential geometry tools such as second fundamental forms, curvatures and Riemannian geometry, while,
as hinted at by Federer's  and Lytchak's results, the weaker positive reach assumption is sufficient for many properties to hold.

For example, in \cite{TangentVarJournal},
the angle variation bound given for $C^2$ manifolds are clearly optimal (the bound is attained on spheres) 
while only suboptimal bounds are derived for $C^{1,1}$, or positive reach, manifolds.

\subparagraph{Attaining the reach}  
In \cite[Theorem 3.4]{aamari2019estimating}, Aamari et al. make the following observation (for $C^2$ manifolds): 
The reach of a submanifold of Euclidean space can be expressed as the minimum of a global quantity, 
realized at so called ``bottlenecks'', and a local quantity, the inverse of the maximal extrinsic curvature. 
However, while this assumption is somewhat implicit in the proof, the analysis uses the context of Riemannian geometry, assuming the $C^2$ regularity of the manifold.

\paragraph*{Motivation and related work}
The assumption of positive reach is central to almost all triangulation algorithms, see for example \cite{ChengDeyShewchuck, Dey, JDbook}, as well as manifold learning, see for example \cite{FeffermanFitting, fefferman2019GeometricWhitneyProblem, Eddie2018stability, sober2019manifold}. 
Closing the gap  between $C^2$ and positive reach manifolds significantly extends the applications domains.
This extension is necessary to include for example the boundary of objects that have been designed by computer aided design software. These boundaries are generically tangent continuous and have bounded curvature (e.g. two planar faces connected through a cylindrical fillet surface), but don't have smooth curvature. 
%

The main result of this paper makes precise what it means for a topologically embedded submanifold of Euclidean space to have positive reach: 
The embedding is necessarily $C^{1,1}$-smooth (Theorem \ref{PosReachImpliesOptimalSizeNeighbourAsGraph}) 
and
we  give a tight bound on the (generalized) extrinsic curvature of the submanifold (Theorem \ref{ReachImpliesQuantifiedLipschitzDerivative}).
 
\subparagraph{Generalized tangent and normal spaces} Federer showed that although sets of positive reach are more general than smooth manifolds, they still posses generalized tangent and normal spaces, which are convex cones instead of just linear spaces.  
Roughly speaking, the normal cone  at a point $p$ of a set with  positive reach determines the topological ``link'' of point $p$ in the set, that is the local topology.

To prove the Theorem \ref{PosReachImpliesOptimalSizeNeighbourAsGraph} 
it is essential to show that the normal cone at any point of a positive reach, topologically embedded $n$-manifold is a $(d-n)$-dimensional vector space. 
In the proofs of \cite{lytchak2004geometry, lytchak2005almost} the relation between the normal cone at a point $p$ of a topologically embedded manifold with positive reach and the ``link'' of $p$ 
is established using CAT(k)-
theory.
In contrast, the short proof of  Lemma \ref{lemma:DefRetractOfBallMinusNormal} 
based on elementary properties of 
(the homology of) convex cones 
makes this correspondence more transparent in our opinion. 

We further extend the results of  \cite{lytchak2004geometry, lytchak2005almost} 
by giving an optimal bound on the angle variation of tangent spaces. This also extends the same bound obtained in the particular context of $C^2$ manifolds \cite{TangentVarJournal}.



\subparagraph{Contribution}

In this paper we give an elementary proof of the following {characterizations of manifolds of positive reach}. 
Moreover the statement is quantified with optimal constants.  
%
Here we write $B^\circ(p, r)$ for the open ball centred at $p$ with radius $r$.  
\begin{restatable}{theorem}{RPosReachImpliesOptimalSizeNeighbourAsGraph}
\label{PosReachImpliesOptimalSizeNeighbourAsGraph}
If $\M$ is a topologically embedded manifold with positive reach {and $p$ a point in $\M$,}
 then $\M \cap B^\circ(p, \sqrt{2}\: \rch(\M)) \cap \pi_{T_p\M}^{-1} (B^\circ(0,\rch(\M)))$ is
the graph of a locally $C^{1,1}$ function above the open domain $B^\circ(0,\rch(\M)) \subset T_p\M$.
\end{restatable}



{ We also generalize the notations of local and global reach from $C^2$ submanifolds \cite{aamari2019estimating}  to arbitrary compact sets (of positive reach):}
\begin{restatable}[Local and global reach]{definition}{LocalGolobalReach}
\label{def:LocalGolobalReach}
For a closed set $\Su \subset \R^d$, we define the local $\rch_{loc.}(\Su)$ and global $\rch_{glob.}(\Su)$ reach 
as:
\begin{align}
&\rch_{glob.}(\Su)  \defunder{=} \min \left \{ \frac{|a-b|}{2} \, \Big | \, {\exists} \,  a,b \in \Su, B^\circ \left( \frac{a+b}{2} , \frac{|a-b|}{2} \right) \cap \Su = \emptyset \right \}  \\
\nonumber \\
&\rch_{loc.}(p, \Su) \defunder{=} \lim_{\substack{\rho \rightarrow 0 \\ \rho >0}}  \, \rch\big( \Su \cap B(p, \rho)  \big)  \in [0, +\infty] 
\label{eq:DefLocReach}
\\
&\rch_{loc.}(\Su ) \defunder{=} \inf_{p \in \Su} \, \rch_{loc.}(p, \Su) {\in [0, +\infty]   }
\label{eq:DefLocReach2}
\end{align}
\end{restatable} 
\begin{remark}\label{remark:ReachLocalWellDefined}
If $\rch(\Su)>0$, then $\rho \mapsto \rch\big( \Su \cap B(p, \rho)  \big) $ is non-increasing (this follows from Theorem \ref{theorem:ReachEquivalentMetricDistorsion} because $\Su \cap B(p, \rho)$ is geodesically convex \cite[Corollary 1]{TangentVarJournal}  for all $\rho< \rch(\Su)$ so that the set of bounds that need to be satisfied in \eqref{eq:MetricDefREach} is larger, c.f. \cite[Lemma 5]{attali2015geometry}) as soon as $\rho< \rch(\Su)$, so that the limit 
in \eqref{eq:DefLocReach} exists. Moreover, since for any $p$ one has  $\rch_{loc.}(p,\Su ) \geq \rch(\Su )$, we get that:
\begin{equation}\label{leq:LocalReachLargerThanReach}
\rch_{loc.}(\Su ) \geq \rch(\Su ).
\end{equation}
\end{remark}

To be able to give the bounds on the tangent variation/extrinsic curvature we need to introduce some notation. 

We write $\Grass(n, \R^d)$ for the Grassmannian, that is the space of $n$-linear subspaces of $\R^d$ with metric given by the sine of the maximum angle. We recall that the maximum angle is given by 
\begin{equation}\label{equation:DefinitionAngleVectorSpaces}
\angle A, B \defunder{=} \max_{a\in A\setminus \{0\}} \min_{b \in B \setminus \{0\}}
 \angle a,b = \max_{b \in B\setminus \{0\}} \min_{a \in A \setminus \{0\}} \angle a,b .
\end{equation} 

 We use this maximal angle as our metric. It is not that difficult to establish that this is indeed a metric, see \cite{wedin2006angles} or \cite{Qiu} for a far reaching generalization.\footnote{The sine of this angle also defines a metric on the Grassmannian as can be found in e.g. \cite[Chapter IV, Section 2.1]{kato2013perturbation} or \cite[Section 34]{akhiezer2013theory}, see Appendix  \ref{app:Grassmann} for a more extensive description.
It seems that in a number of fields the sine of the angles is a more common metric, but it is suboptimal in this context. 
We further stress that neither of these metrics does not equal the standard Riemannian metric on the Grassmannian, which is given by the sum of the squares of the principal angles, see e.g. \cite{wong1967differential}.}


\begin{restatable}{theorem}{ReachImpliesQuantifiedLipschitzDerivative}
\label{ReachImpliesQuantifiedLipschitzDerivative}
If $\M \subset \R^d$ is a topologically embedded $n$-manifold with reach at least $R>0$, then tangent spaces and normal spaces 
to $\M$  are linear spaces at any point  $p \in \M \setminus  \partial \M$ and the maps
 $\Tan: \M  
\rightarrow \Grass(n, \R^d)$ and 
$\Nor: \M 
\rightarrow   \Grass(d-n, \R^d)$ 
 are 
Lipschitz. 
 
Moreover, 
the smallest Lipschitz constant satisfied by $p \mapsto \Tan(p,\Su)$ is $\frac{1}{\rch_{loc.}(\M)}$:
\begin{equation}\label{eq:TightTgtVariation}
\sup_{ \myatop{p,q\in \M}{p\neq q} } \frac{\angle \Tan(p,\M),  \Tan(q,\M)}{d_{\M}(p,q)} = \frac{1}{\rch_{loc.}(\M)},
\end{equation}
where we say that $1/\infty =0$ {and that $d_{\M}(p,q)= \infty$ when $p$ and $q$ are not connected by any path in $\M$}.
\end{restatable}

\begin{remark} 
If $\M \subset \R^d$ is a manifold with boundary then the results of Theorem \ref{ReachImpliesQuantifiedLipschitzDerivative} still hold as long as there is a geodesic between $p$ and $q$ that lies completely in the interior of $\M$. 
\end{remark} 

Next Theorem is similar to \cite[Theorem 3.4]{aamari2019estimating}  stated in the particular situation where 
$\Su$ is a $C^2$ embedded manifold, in which case $\rch_{loc.}(\Su)$ is also the supremum of extrinsic curvature.

\begin{restatable}{theorem}{ReachEquivalentMinLocalGlobal}\label{theorem:ReachEquivalentMinLocalGlobal}
If $\Su$ is a  compact subset of $\R^d$, then:
\[
\rch( \Su )= \min ( \rch_{glob.}(\Su), \rch_{loc.}(\Su) )
\]
\end{restatable}

{Full proofs of the statements that are not in the main paper can be found in the appendix, but we have also added proof sketches for the non-straightforward statements. }

\section{Preliminaries}\label{sec:Preliminaries} 
In this section we first recall some results and definitions concerning sets of positive reach. The results concerning the reach are mainly taken from \cite{Federer} and \cite{TangentVarJournal}. 

\subsection{Sets of positive reach}
	\begin{tcolorbox}Throughout this paper we will write $\Su$ for a general closed set (mostly with positive reach) and $\M$ for a manifold (mostly without boundary) and in some rare cases $\partial \M$ for its boundary if it exists. 
	\end{tcolorbox}

The concept of a tangent space of a smoothly embedded manifold has been generalized by Federer to tangent cones in the following manner: 
\begin{definition}[Generalized tangent spaces, Definitions 4.3 and 4.4 of \cite{Federer}] \label{def:4.3and4.4Fed} 
If $\Su \subset \mathbb{R}^d$ and $p\in \Su$, then the generalized tangent space
\begin{align}
\Tan(p,\Su)
\nonumber
\end{align}
is the set of all tangent vectors of $\Su$ at $p$ consists of all those $u \in \R^d$, such that either $u=0$ or for every $\epsilon>0$ there exists a point $q \in  \Su$ with
\begin{align}
0<&|q-p|<\epsilon &\textrm{and} & &\left| \frac{q-p}{|q-p|}- \frac{u}{|u|} \right| < \epsilon.
\nonumber 
\end{align}
The set 
\begin{align}
\Nor(p,\Su)
\nonumber
\end{align}
of all normal vectors of $\Su$ at $p$ consists of all those $v \in \R^d$ such that  $v \cdot u \leq 0$ for all $u \in \Tan(p,\Su)$.
\end{definition}


\begin{definition} 
We write $\pi_\Su$ for the closest point projection on $\Su$. 
\end{definition} 
Although it is possible to define a set valued version of $\pi_\Su$, we only use the projection map outside the medial axis, that is at those points where there is a single closest point. 

For the remainder of this subsection we will recall the properties of this map and sets of positive reach that are essential for this paper.


\begin{lemma}[Angle with tangent space, Theorem 4.8(7)  of \cite{Federer}] \label{Lem:distanceToTangent} 
Suppose that $\Su$ is a set positive reach $\rch(\Su)$.
Let $p,q \in \Su \subset \mathbb{R}^d$ such that $|p-q| < \rch(\Su)$. We have 
\begin{align} 
\sin \angle (q-p, \Tan(p, \Su) ) \leq \frac{|p-q|}{2\,  \rch (\Su)}.
\label{AngleSegmentToT}
\end{align} 
\label{TangentSpaceLine1}
\end{lemma}


Moreover, a variant of the previous lemma characterizes the reach for arbitrary closed sets, that is:
\begin{lemma}[Distance to tangent space characterizes the reach, Theorem 4.18  of \cite{Federer}] 
\label{Lem:distanceToTangentCharacter} 
For a closed subset $\Su$ of Euclidean space $\mathbb{E} = \R^d$, one has:
\begin{equation}\label{TangentSpaceLine_2}
\rch(\Su) = \inf_{\myatop{p,q\in \Su}{p\neq q}} \frac{|p-q|^2}{2 d_{\mathbb{E} } (q ,\: p \: + \Tan(p, \M) ) } 
\end{equation}
\end{lemma}

\begin{lemma}[Theorem 4.8 (4) of \cite{Federer}]\label{Fed4.8.4}
$\pi_\Su$ is continuous on $\mathbb{R}^d  \setminus \ax(\Su)$.
\end{lemma}
In fact the following lemma strengthens the previous. 

\begin{lemma}[Theorem 4.8 (8) of \cite{Federer}]\label{Fed4.8.8}
If $x,y \in \mathbb{R}^d$, and $\max \{|x - \pi_{\Su}(x)| , |y - \pi_{\Su}(y)|\} = \mu \leq \rch(\Su)$, then
\begin{align} 
|\pi_\Su(x)- \pi_{\Su}(y) |\leq \frac{ \rch(\Su)}{ \rch(\Su)- \mu} |x-y| 
\nonumber
\end{align} 
\end{lemma}

\begin{lemma}[Theorem 4.8 (12) of \cite{Federer}]\label{Fed4.8.12}
If $p \in \Su$ and $\lfs (p)> \rho > 0 $, then 
\begin{align}
\Nor ( p,\Su) = \{ \lambda v \mid \lambda\geq 0 , | v| = \rho ,  \pi_\Su (p+v) =p\} .
\nonumber
\end{align}
$\Tan(p,\Su)$ is the convex cone dual to $\Nor(p,\Su)$, and 
\begin{align}
\lim_{t \to 0^+} t^{-1} d_\M (p +t u) =0,
\nonumber
\end{align}
for $u\in T(p,\Su)$ 
\end{lemma}

\begin{lemma}[Remark 4.15 (1) of \cite{Federer}]\label{Fed4.15.1}
If $p \in \Su$ and $\lfs (p)> \rho > 0 $, then $\Su \cap B(p, \rho)$ is contractible.
\end{lemma}


\subsubsection{Geodesics}\label{section:geodesics}

We also recall a result from the second paragraph of Part III, Section 1: ``Die Existenz geod{\"a}tischer Bogen in metrischen R{\"a}umen'' in \cite{menger1930untersuchungen}:
\begin{lemma}[Menger's existence of geodesics] \label{Menger}
For a closed set $\Su \subset \R^d$, $d_{\Su}$ denotes the geodesic distance in $\Su$, i.e. $d_{\Su}(a,b)$ is  the infimum of lengths of paths in $\Su$ between
$a$ and $b$. 
If there is at least one path   between  $a$ and $b$ with finite length, then it is known that a minimizing geodesic, i.e. a path
with minimal length connecting $a$ to $b$ exists. 
When no such path  between  $a$ and $b$ with finite length exists, we write $d_{\Su}(a,b) = \infty$. 
\end{lemma}
We'll be using the metric characterization of the reach from \cite{TangentVarJournal} quite often in this paper. 
\begin{theorem}[Theorem 1 of \cite{TangentVarJournal}]\label{theorem:ReachEquivalentMetricDistorsion}
If $\Su \subset \R^d$ is a closed set, then
\begin{align} \label{eq:MetricDefREach} 
\rch \, \Su = \sup \left\{r > 0,  \, \forall a,b \in \Su, \, |a-b| < 2 r \Rightarrow d_{\Su}(a,b) \leq 2 r \arcsin \frac{|a-b|}{ 2 r} \right\},
\end{align} 
where the $\sup$ over the empty set is $0$.
\end{theorem}

\subsection{Embeddings and atlases} 
We consider below an embedded, compact $n$-manifold $\M \subset \R^d$ with an atlas associated with an open cover $(U_i)_{i\in I}$:
\[
\M = \bigcup_{i\in I} U_i 
\]
and charts $f_i: U_i \rightarrow \R^n$ that are homeomorphisms on their images, where the topology on $\M$ is the one induced by the ambient metric of $\R^d$.

Note that, topologically embedded manifold does not exclude wild embeddings \cite{daverman2009embeddings}  but, as seen below,
the positive reach assumption does.


\section{Convex cones}\label{section:convexCones}

In this section we recall the definition of a convex cone, and some known consequences. We follow Chapter 2 of \cite{rockafellar1970convex} and refer to this book for further reading. After the short recap, we characterize vector spaces among convex cones by means of their homology, which we have not been able to find in the literature. 

In one of the equivalent definitions of a convex cone we use geodesic convexity on the sphere $\Sphere^{d-1}$, which we'll recall first. 
\begin{definition}\label{definition:geodesicallyConvexOnSphere}
A set of directions $S \subset \Sphere^{d-1}$ is said 
 {\em geodesically convex}  if, given $p,p'\in S$ such that $\angle p,p'  < \pi$,  the unique minimizing geodesic joining $p$ and $p'$,
i.e. the arc $pp'$ of a great circle strictly shorter than $\pi$, is contained in $S$.
\end{definition}
The notion of convex cone is standard. We use the following definitions, the first of which is standard and the second is easy to verify. 
\begin{definition}[Convex cone]\label{definition:convexCone}
For a subset $A \subset \R^d$, the following properties are equivalent (and define convex cones): 
\begin{enumerate}[label=\textnormal{(\arabic*)}]
\item \label{item:DefConvexCone1} 
$\displaystyle 
\forall a_1,a_2 \in A, \forall \lambda_1, \lambda_2 \geq 0,\: \lambda_1 a_1 +\lambda_2 a_2 \in A
$
\item 
Both 
\[
\forall a \in A, \forall \lambda \geq 0,\:  \lambda  a  \in A
\] and 
\[
A \cap\Sphere^{d-1} \quad \textrm{is geodesically convex in} \quad  \Sphere^{d-1}
\]
\end{enumerate}
\end{definition}

If $A \subset \R^d$, $A^\perp$, its dual convex cone is 
\[
A^\perp \defunder{=} \{ x \in \R^d, \forall a \in A, \langle a, x \rangle  \leq 0 \}. 
\]
We have the following, see \cite{rockafellar1970convex}. 
\begin{lemma}\label{lemma:DualOfMyDualIsMyself}
If $A$ is a closed convex cone it is the dual of its dual: $A = (A^\perp)^\perp$ .
\end{lemma}

Also, from Definition \ref{definition:convexCone}, item \ref{item:DefConvexCone1}, we see that:
\begin{observation}\label{observation:ConvexConeSymetricIFFVectorSpace}
For a convex cone $A$, we have that $A=-A$ if and only if $A$ is a vector space.
\end{observation}
\begin{observation} 
\label{Obs:dualVspace}
A convex cone $A$ is a $k$-dimensional vector space if an only if $A^\perp$ is a $(n-k)$-dimensional vector space.
\end{observation}

Moreover vector spaces can be distinguished among convex cones by the homology of their
 intersection with the unit sphere $\Sphere^{d-1}$.
More precisely, we have:
\begin{lemma}\label{lemma:HomologyCCone1}
For a convex cone $A\subset \R^d$ and a non negative integer $k\leq d$, the following are equivalent:
\begin{itemize}
\item $A$ is a $k$-dimensional vector space,
\item $k\geq 1$ and the $(k-1)$-homology of $A \cap \Sphere^{d-1}$ is non-trivial or $k=0$ and $A \cap \Sphere^{d-1} = \emptyset$.
\end{itemize}
\end{lemma}

Vector spaces can also be distinguished among convex cones by the homology of the  complement of their dual cone in the unit ball.
More precisely, we have:
\begin{lemma}\label{lemma:HomologyCCone2}
Let $A\subset \R^d$ be a closed convex cone and let $k\leq d$ be a non negative integer. 
The following are equivalent:
\begin{itemize}
\item The cone $A$ is a $k$-dimensional vector space.
\item If $k>0$, the $(k-1)$-homology of $B(0,1) \setminus  A^\perp$  is non-trivial and if $k=0$, $B(0,1) \setminus  A^\perp$ is empty.
\end{itemize}
\end{lemma}

{
\begin{proof}[Main idea of the proofs of Lemmas \ref{lemma:HomologyCCone1} and \ref{lemma:HomologyCCone2}]  
The key observation for both Lemma \ref{lemma:HomologyCCone1} and Lemma \ref{lemma:HomologyCCone2} is that a cone is not a vector space then the intersection with a sphere is geodesically star shaped. 
\end{proof}}

\section{The (generalized) tangent spaces of embedded manifolds}
In this section we establish that the tangent spaces of topologically embedded manifolds are vector spaces of the right dimension. The argument is based on homotopy theory and deformations retract in particular, which uses a more general statement that we proof first. 

\subparagraph{Closest point projection in a neighborhood gives a deformation retract}
We know  (Lemma \ref{Fed4.15.1} above and Lemmas 5 and 9 in \cite{attali2015geometry}), that, if $\rho < \rch(\Su)$,  
then $\rch \left( \Su \cap B(p, \rho) \right) \geq \rch(\Su)$ and that $\Su \cap B(p, \rho)$
is contractible. Moreover one has:

\begin{lemma}\label{lemma:DefRetractOfBallMinusNormal}
If $0< \rho < \rch(\Su)$, the projection $\pi_{\Su \cap B(p, \rho)}$ defines a deformation retract from $B(p, \rho) \setminus \Big( B(p, \rho) \cap \big( p + \Nor(p, \Su) \big) \Big)$ to $\Su \cap B(p, \rho) \setminus \{p\} $.
\end{lemma}

{\begin{proof}[Sketch of proof] This follows by combining a number of results of Federer, in particular the continuity of the closest point projection. 
\end{proof}}

\subparagraph{Tangent cones of embedded \texorpdfstring{$\mathbf{n}$}{n}-manifolds with positive reach are \texorpdfstring{$\mathbf{n}$}{n}-vector spaces}

 \begin{lemma}\label{lemma:embeddedManifoldPositiveReachThenTangentSpaceIsVectorial}
Let $\M \subset \R^d$ be a topologically embedded  $n$-manifold in $\R^d$, with positive reach.
Then for any $p \in \M$, $\Tan(p,\M)$ is a $n$-dimensional vector space.
 \end{lemma}

\begin{proof}
By definition of an embedded manifold, $\M$ has the induced topology. 
Let $p\in U$, where $U \subset \M$ is the open domain of a chart in the atlas of $\M$. This means that there exists a map $f: U \rightarrow \R^{n}$ that is a homeomorphism on its image.
We note that because $\M$ is compact, or equivalently closed and bounded, $\M \setminus U$ is compact too and thus $d(p,  \M \setminus U) >0$. 
We now choose $\rho$ 
\[
 0 < \rho < \min \Big(\rch( \M), d(p,  \M \setminus U) \Big).
\]
Note that in particular $\M \cap B(p, \rho) \subset U$.

By Lemma \ref{Fed4.15.1} the set $\M \cap B(p, \rho)$ is contractible. 
Because $f$ is an homeomorphism on its image $f( \M \cap B(p, \rho))$ is a contractible subset of $\R^n$ containing a neighborhood of $f(p)$.
Therefore, we know that the $n-1$ homology of $f( \Su \cap B(p, \rho)) \setminus \{f(p)\} $ is non-trivial. 
Using again that $f$ is an homeomorphism on its image, we see that the set $f( \Su \cap B(p, \rho)) \setminus \{f(p)\} $ is homeomorphic to $\Su \cap B(p, \rho) \setminus \{p\}$. By Lemma \ref{lemma:DefRetractOfBallMinusNormal},  $\Su \cap B(p, \rho) \setminus \{p\}$ has the same homotopy type as 
$B(p, \rho) \setminus \left( p + \Nor(p, \Su) \right) $.

We get that  the $n-1$ homology of $B(p, \rho) \setminus \left( p + \Nor(p, \Su) \right) $ is non-trivial, which, by Lemma 
\ref{lemma:HomologyCCone2} gives 
that  $\Tan(p, \M)$ is a $n$-dimensional vector space.
\end{proof}

\section{ Semi-continuity of normal cones and \texorpdfstring{$\mathbf{C^1}$}{once differentiable} embedding}\label{section:C1Embedding}

In this section we first discuss a weak form of continuity for tangent and normal cones of general sets of positive reach before focusing on manifolds.  

\subparagraph{Semi-continuity of $p \mapsto \Nor(p, \Su)$ and $p \mapsto \Tan(p, \Su)$ for sets of positive reach}

For a cone $X$, $X^{\angle \epsilon}$ is the set of vectors making an angle less than $\epsilon$ with some vector in $X$.
For sets with positive reach, we have a semi-continuity of the normal and tangent cones, more precisely,
\begin{lemma} \label{lemma:NorIsSemiContinuous}
If $\rch(\Su) >0$, then
\begin{equation}\label{equation:SemiContinuityNormal}
\forall p \in \Su, \forall \epsilon >0, \exists \alpha >0 \mid \:  p' \in \Su \cap B(p, \alpha) \Rightarrow \Nor(p', \Su) \subset \Nor(p, \Su)^{\angle \epsilon}.
\end{equation}
\end{lemma}

\begin{proof}[Sketch of proof] The core of the proof is rather technical, but requires only basic topological properties, and relies on the manipulation of offsets. 
\end{proof}

The semi-continuity of normal cones translates into a symmetric result for tangent cones. 
This is easy to see: For set $\Su$ with positive reach, 
$\Tan(p, \Su) = \Nor(p,\Su)^\perp$, and, if $X$ and $Y$ are convex cones, then:
\begin{align}
 X\subset Y^{\angle \epsilon} &\iff   \forall x \in X \setminus \{0\},    \forall y' \in Y^\perp \setminus \{0\}, \: \angle x, y' > \frac{\pi}{2}  - \epsilon   \nonumber \\
 & \iff Y^\perp \subset \big(X^\perp\big)^{\angle \epsilon}.
\nonumber
\end{align}
Lemma \ref{lemma:NorIsSemiContinuous} therefore similarly (but inclusion reversed) gives semi-continuity of $p \mapsto \Tan(p, \Su)$
\begin{equation}\label{equation:SemiContinuityTangent}
 \forall p \in \Su, \forall \epsilon >0, \exists \alpha >0 \mid \:  p' \in \Su \cap B(p, \alpha) \Rightarrow \Tan(p, \Su) \subset \Tan(p', \Su)^{\angle \epsilon}.
\end{equation}

\subparagraph{Continuity if the tangent spaces are vector spaces}

Of course, in cases where $\Tan$ and $\Nor$ are vector spaces, this semi-continuity gives continuity. 
This results from a straightforward lemma (we included a proof in the appendix for completeness):
\begin{lemma}\label{lemma:anglebetweenKVectorSpacesIsSymmetric}
If $A$ and $B$ are $k$-dimensional vector subspaces of $\R^d$, with $1\leq k \leq d-1$, then
\[
A \subset B^{\angle \epsilon} \iff B \subset A^{\angle \epsilon}.
\]
\end{lemma}

From Lemma \ref{lemma:anglebetweenKVectorSpacesIsSymmetric}, we get that, in case of manifolds, 
Lemma \ref{lemma:NorIsSemiContinuous} and equation \eqref{equation:SemiContinuityTangent}
gives us the continuity of tangent spaces.
More precisely, we have:

\begin{lemma}\label{lemma:embeddedManifoldPositiveReachThenC1Embedded}
If $\M$ is a topologically embedded $n$-manifold in $\R^d$, with positive reach,
then $\M$ is $C^1$ embedded in $\R^d$. 
More precisely for any $p\in \M$, there exists some open $p \in U_p \subset \Tan(p, \M)$ and a map $\phi: U_p \rightarrow \Nor(p, \M)$, such that $x \mapsto  \Phi(x) \defunder{=} p + x + \phi(x)$ is a $C^1$
map from $U_p$ to $\M$.
 \end{lemma}

{
\begin{proof}[Sketch of proof] The proof gathers the results so far and then uses standard techniques from analysis and geometry to find locally find a differentiable map from the tangent space to the normal space. 
\end{proof}
}

%
%
\section{From \texorpdfstring{$\mathbf{C^1}$}{once differentiable} to \texorpdfstring{$\mathbf{C^{1,1}}$}{a Lipschitz derivative}}

\label{section:C11Embedding}

In this section we prove the first main result of this paper, namely that manifolds of positive reach are $C^{1,1}$. We distinguish ourselves from the previous methods by making all constants explicit. 


%

{In this paper we view derivatives of a $C^{1,1}$ function $F$ as linear operators (matrices). The norm we use on the derivatives of $F$ is the operator $2$-norm, that is 
\[ \|  DF(x)\| = \max_{ |u|=1} \frac{|DF(x) u | }{|u|} . 
\]
For example, $F$ having a $\lambda$-Lipschitz derivative means that the difference of Jacobians viewed as matrices (or linear operators) satisfies :
\begin{equation}
\| DF (x)-DF (y) \| \leq \lambda |x-y|\: \text{ for all} \:  x,y. 
\end{equation}
  }

The following Lemma on $C^{1,1}$ functions is essential to our proof and can be found in a number of texts concerning semiconcave functions  as well as \cite{lytchak2004geometry, lytchak2005almost}, albeit without constants, the version with constants we recall here can be found as Corollary 3.3.8 in \cite{cannarsa2004semiconcave}. 
\begin{lemma}\label{lemma:SemiConcaveAndSemiConvexIffC11}
Let $F: U \rightarrow \R$ be a function defined on a convex open set $U \subset \R^n$.
Then the following properties are equivalent:
\begin{itemize}
\item $F$ is of class $C^{1,1}$ with $1$-Lipschitz derivative.
\item both $x \rightarrow 1/2 \: x^2 - F(x)$ and $x \rightarrow 1/2 \:  x^2 + F(x)$ are convex functions.
\end{itemize}
\end{lemma}



Following two lemmas are of use in the proof of Lemma \ref{lemma:VDotphiIsSemiConvex}.
\begin{lemma}\label{lemma:SmallLemma}
Assume that $\M \subset \R^d$ is a topologically embedded $n$-manifold, with  reach larger or equal to $R>0$. Suppose that $p\in \M$, and let $U_p\subset \Tan(p, \M)$, $\phi: U_p \rightarrow \Nor(p, \M)$, and $\Phi$ be as in Lemma \ref{lemma:embeddedManifoldPositiveReachThenC1Embedded}.

Then, for all $\theta >0$, there is an $\alpha >0$ 
such that:
 \begin{equation}
|y|, |y'| \leq \alpha 
\Rightarrow 
\begin{cases}
&\angle \Tan (\Phi(y), \M) ,  \Tan (\Phi(y'), \M)< \theta  
\\
 &\operatorname{and}\\
&\angle \Phi(y') - \Phi(y), \Tan(p, \M) < \theta.\label{eq:BoundYTheta} 
\end{cases}
 \end{equation}
\end{lemma}
The proof is straightforward and can be found in the appendix. 


\begin{lemma}\label{lemma:ProjectionOfTangentConeBoundaryEqualProjOnTgtSpace}
For any $p\in \M$, there is $\alpha>0$ such that for any $\alpha'$ such that $0<\alpha' \leq \alpha$,
for any $q \in \M \cap  \partial \big(B(p, \alpha' ) \big)$ and 
 $q' \in \M \cap B(p, \alpha' )$ the closest point projection of $q'$ on 
 $q \, + \Tan(q, \M)$ belongs to ${ q \, +} \Tan(q, \M \cap B(p, \alpha' ) )$, in particular, one has:
\begin{align}\label{eq:claim_intersect_tangent_to_manifold}
d_{\mathbb{E} } (q' ,\, q \, + \Tan(q, \M \cap B(p, \alpha' ) )  = d_{\mathbb{E} } (q' ,\, q \, + \Tan(q, \M) ) 
\end{align}
\end{lemma}
The result in Lemma \ref{lemma:ProjectionOfTangentConeBoundaryEqualProjOnTgtSpace} is rather intuitive but the proof requires some geometric constructions and can be found in the appendix. 


\begin{lemma}\label{lemma:VDotphiIsSemiConvex}
Assume that $\M \subset \R^d$ is a topologically embedded $n$-manifold with  positive reach. 
Suppose that $p\in \M$, with $\rch_{loc.}(p, \M) >0$
and let $U_p\subset \Tan(p, \M)$, $\phi: U_p \rightarrow \Nor(p, \M)$, and $\Phi$ be as in Lemma \ref{lemma:embeddedManifoldPositiveReachThenC1Embedded}.
 
Then, for all $\epsilon >0$ and $v \in \Nor(p, \M)$,  with $| v | = 1$, there is an $\alpha >0$ such that the restriction to $U_p \cap B(0, \alpha)^\circ$ of
  \[
  x \mapsto 1/2 \: x^2 - (\rch_{loc.}(p, {\M}) - \epsilon)  \langle v, \phi(x) \rangle
  \]
  is convex.
  {
If $\rch_{loc.}(p, {\M}) = \infty$, then the statement holds with $\rch_{loc.}(p, {\M})$ replaced by any $R>\epsilon $.   
}	

 Moreover,  if $r >  \rch_{loc.}(p, {\M})$
then, for any $\eta>0$, there is  $v \in \Nor(p, \M)$,  with $| v | = 1$,
such that the restriction of 
\[
 x \mapsto 1/2 \: x^2 -  r  \langle v, \phi(x) \rangle
\]
to $B(0,\eta)$ is not convex. 
\end{lemma}

{
\begin{proof}[Sketch of proof] The proof of Lemma \ref{lemma:VDotphiIsSemiConvex} is rather intricate, but boils down to upper bound some angles in Figure \ref{fig:GraphOfSemiConvex}. 
\end{proof}}  



A relatively straightforward consequence of Lemmas \ref{lemma:SemiConcaveAndSemiConvexIffC11} and \ref{lemma:VDotphiIsSemiConvex}
is the following (a complete proof of the statement has been added in the appendix for completeness):

\begin{lemma}\label{lemma:phiIsC11}
  Consider a $n$-manifold $\M \subset \R^d$ topologically embedded in $\R^d$, with  positive reach,
a point  $p\in \M$,    $U_p\subset \Tan(p, \M)$ a neighborhood of $p$  in $\Tan(p, \M)$,
and $\phi: U_p \rightarrow \Nor(p, \M)$, the a $C^1$ map defined in Lemma \ref{lemma:embeddedManifoldPositiveReachThenC1Embedded},
  such that $x \mapsto \Phi(x) = p + x + \phi(x)$ is a chart from $U_p$ to $\M$.
  
  Then, for all $\epsilon >0$, there is $\alpha >0$ such that the derivative of the restriction of $\phi$  to $U_p \cap B(0, \alpha)^\circ$
   is $\frac{1}{(\rch_{loc.}(p, \Su) - \epsilon)}$-Lipschitz,
  where the distance on derivatives is the  operator norm of the difference, that is 
	\[
\big\|D\phi(y_2)  - D\phi(y_1)  \big\|  \leq \frac{1}{(\rch_{loc.}(p, \Su)  - \epsilon)}  | y_2 - y_1 |.
\]
If $\rch_{loc.}(p, {\M}) = \infty$, then the statement holds with $\rch_{loc.}(p, {\M})$ replaced by any {$R>\epsilon$}.   	
 \end{lemma}

Next lemma is instrumental in the proof of Lemma \ref{lemma:MIsC11WithOptimalBounds}:
\begin{lemma}\label{lemma:FromOperatorNormToAngles}
For $F_1,F_2$ linear maps from $\R^n$ to $\R^m$,
denote by  $L_1$ and $L_2$ the corresponding graphs in the product space $\R^n \times \R^m =\R^{n+m}$
 equipped with the canonical dot product: $L_i = \{ (x, F_i(x) ) \in \R^{n+m} \}$.
 \[
2 \sin  \frac{\angle L_1, L_2}{2} \leq   \big\| F_2 - F_1 \big\|
\]
where $  \| \cdot    \|$ is the $L^2$ operator norm.
\end{lemma}
\begin{proof}[Sketch of proof] 
The proof of Lemma \ref{lemma:FromOperatorNormToAngles} involves a number estimates on angles, which then are relatively straightforwardly reformulated in terms of an operator norm. 
\end{proof}

\begin{lemma}\label{lemma:MIsC11WithOptimalBounds}
If $\M \subset \R^d$ is a $n$-manifold  topologically embedded in $\R^d$, with positive reach,
then $\M$ is $C^{1,1}$ embedded 
 and if  $p_1, p_2 \in \M$ then
 \[
 \angle \Tan(p_1, \M), \Tan(p_2,\M) \leq  \frac{1}{ \rch_{loc.}(\M) }  d_{\M} (p_1, p_2),
\]
where $d_{\M}$ denotes the geodesic distance inside $\M$ and we say that $1 / \infty  =0$.

If now $\M$ is an $n$-manifold  with boundary, topologically embedded in $\R^d$, the same property holds when the geodesic between $p_1$ and $p_2$
avoids the boundary.  
  \end{lemma}

\begin{proof}[Sketch of proof] 
The proof consists in gluing the local estimates on the angle between tangent spaces that follow by combining Lemmas \ref{lemma:phiIsC11} and \ref{lemma:FromOperatorNormToAngles} along a geodesic. 
\end{proof}


We can now prove one of our main theorems:
\ReachImpliesQuantifiedLipschitzDerivative*

\begin{proof}[Sketch of proof] One direction of the proof is given by Lemma \ref{lemma:MIsC11WithOptimalBounds}, the inverse statement relies on the second statement of Lemma \ref{lemma:VDotphiIsSemiConvex} together with a significant amount of computation. 
\end{proof}

\begin{lemma}\label{lemma:LocalReachAsTangentVaraiationBoundOnGeodesics}
 If $\M \subset \R^d$ is a $n$-manifold  topologically embedded in $\R^d$, 
 and $\gamma$ a geodesic in $\M$ of length $\ell$, the local reach of 
 $\gamma([0,\ell])$ is upper bounded by the local reach of $\M$ and in particular:
\begin{align}\label{eq:TightTgtVariationGeodesics}
&\forall s_1,s_2 \in [0,\ell ], &&\angle \dot{\gamma}(s_1),  \dot{\gamma}(s_2) \leq \frac{1}{\rch_{loc.}(\M)} |s_2 - s_1 |
\end{align}
 \end{lemma}

\begin{proof}[Sketch of proof] The lemma follows by combining Theorem \ref{theorem:ReachEquivalentMetricDistorsion} and Lemma \ref{lemma:MIsC11WithOptimalBounds}.  
\end{proof}

\RPosReachImpliesOptimalSizeNeighbourAsGraph*

{
\begin{proof}[Sketch of proof] By combining Theorems \ref{ReachImpliesQuantifiedLipschitzDerivative} and \ref{theorem:ReachEquivalentMetricDistorsion} with the submersion theorem we can establish a local diffeomorphism. The global diffeomorphism is established by proving that the covering number is one. 
\end{proof}

}




Lemma \ref{cor:OtherPaper} provides us with an explicit bound on $\alpha$ in Lemmas \ref{lemma:VDotphiIsSemiConvex} and \ref{lemma:phiIsC11}. Such a bound is useful for explicit calculations, see e.g. \cite{FreeLunch}. 
\begin{lemma} \label{cor:OtherPaper} 
As long as 
\[
 \frac{R- \epsilon}{R} \geq 1/2, 
\] or equivalently $\epsilon/ R \leq 1/2$,
it suffices for  $\alpha$, as defined in first statement of Lemmas \ref{lemma:VDotphiIsSemiConvex} and \ref{lemma:phiIsC11},  to satisfy
\[
\alpha \leq \sqrt{\epsilon R} .
\]
\end{lemma} 

The proof of this statement is relatively straightforward and can be found in the appendix. 



\section{Global and local reach}\label{section:GlobalLocalReach}

In this section we generalize the results of \cite{aamari2019estimating} on how the reach is attained from $C^2$ manifolds to arbitrary sets of positive reach. This result states that the reach is a consequence of either curvature maxima or so-called bottlenecks. In our result we replace the curvature by the local reach, which in turn is defined using local metric distortion, while the bottlenecks are similar in nature. In a number of cases we will have to assume that the set is bounded/compact.  

We recall the definition from the introduction:
\LocalGolobalReach*

\begin{lemma}\label{lemma:LocalReachLowerSC} 
The map $p\mapsto \rch_{loc.}(p, \Su)$  defined over $\Su$ is lower semi-continuous, in other words:
\[
\rch_{loc.}(p, \Su) = \lim_{\substack{\rho \rightarrow 0 \\ \rho >0}} \: \inf_{q\in \Su \cap B^{\circ}(p, \rho)} \:  \rch_{loc.}(q, \Su).
\]
\end{lemma}
\begin{proof}[Sketch of proof] The proof is a direct consequence of definition. 
\end{proof}

\begin{remark}
As a consequence, when $\Su$ is compact, the $\inf$ in \eqref{eq:DefLocReach2} is a $\min$, in other words there exists at least one point $p\in \Su$
such that $\rch_{loc.}(p, \Su) = \rch_{loc.}(\Su)$.
\end{remark} 

\begin{definition} A bottleneck is a pair of points $(a,b) \in \Su ^2$ such that 
\[B^\circ \left( \frac{a+b}{2} , \frac{|a-b|}{2} \right) \cap \Su = \emptyset. \] 
\end{definition} 


We have the following variant of Lemma 12 of \cite{TangentVarJournal}, which in turn is a relatively straightforward generalization to arbitrary dimension of Property I of \cite{attali:hal-00201055}:
\begin{lemma} \label{lemma:LocalReachMakeGeodesicsStraight}
If a geodesic (or more generally a curve) $\gamma$ parametrized according to arc length on the interval $[0,\ell]$ satisfies 
\begin{equation}\label{equation:GeodesicTangentLipschitz}
\angle \dot{\gamma} (s_1) ,\dot{\gamma} (s_2) \leq \frac{1}{R}  | s_2 -s_1 |,
\end{equation}
then as long as $\ell \leq {2} \pi R$.
\[
|q-p| \geq 2 \, \rch(\M) \sin \left(\frac{ \ell }{2R} \right),
\]
where $p= \gamma(0)$ and $q=\gamma (\ell)$. 
\end{lemma}

Combining Lemmas \ref{lemma:LocalReachAsTangentVaraiationBoundOnGeodesics} and
 \ref{lemma:LocalReachMakeGeodesicsStraight} we get:
\begin{corollary}\label{lemma:ShortGeodesicsAreStraights}
If $\gamma$ is a geodesic in $\Su$ with length $\ell \leq 2 \pi \rch_{loc.}(\Su)$, then
\[
\ell \leq 2 \rch_{loc.}(\Su) \arcsin \frac{|p-q|}{2 \rch_{loc.}(\Su) }.
\]
\end{corollary}


The Lemma \ref{lemma:SemiContinuitySpecialCase} below is a specialization of 
\cite[Lemma 4.6]{LIEUTIERhomotopytype}, but we still provide a proof in the appendix. 
\begin{lemma}\label{lemma:SemiContinuitySpecialCase}
Let $x \notin \ax(\Su)$ and denote by $\tilde{x}$ 
the unique point   closest to $x$ in $\Su$.
\newline Then, for any $\epsilon>0$, there is $\alpha>0$ such that,
\begin{equation}\label{eq:SemiContinuitySpecialCase}
|y-x| < \alpha \Rightarrow \{ z \in \Su \mid d(y,z) = d(y, \Su) \} \subset B^{\circ}(  \tilde{x} , \epsilon ).
\end{equation}
\end{lemma}

\ReachEquivalentMinLocalGlobal*

\begin{proof}

{
We consider first the situation where $\rch( \Su )=0$.
 In this case, since $\Su$ is bounded, there exists some point $p \in \overline{\ax(\Su)} \cap \Su$.
 Then, for any $\rho>0$, there is some $p' \in \ax(\Su)$ such that $d(p', p) < \rho/2$, so that there are at least two points in 
 $\Su \cap B(p, \rho)$ closest to $p'$. 
 Thus,  $\rch(\Su \cap B(p, \rho)) < \rho/2$, and, since this holds for any $\rho>0$, 
 we get  that $\rch_{loc.}(p,\Su) = 0$ and {as a direct consequence}  $\rch_{loc.}(\Su) = \inf_{q\in \Su} \rch_{loc.}(q,\Su) = 0$.

We assume now that $\rch( \Su )>0$.
}
Thanks to the metric distortion definition we already know thanks to Remark \ref{remark:ReachLocalWellDefined}, Equation \eqref{leq:LocalReachLargerThanReach} that $\rch( \Su ) \leq \rch_{loc.}(\Su)$.
Therefore it suffices to prove that, if $\rch( \Su ) < \rch_{loc.}(\Su)$, then $\rch( \Su )  = \rch_{glob.}(\Su)$.
So let us assume now that $\rch( \Su ) < \rch_{loc.}(\Su)$.

Since $\Su$ is compact and thus bounded, which implies that a closed $\rch(\Su)$-neighbourhood of $\Su $ is also bounded, there must exists $x\in \overline{\ax(\Su)}$,
where $\overline{\ax(\Su)}$ is the closure of the medial axis $\ax(\Su)$,
such that $d(x, \Su) = \rch( \Su )$.

If  $x\in \overline{\ax(\Su)} \setminus  \ax(\Su)$, in particular $x \notin \ax(\Su)$ and
 there is a unique point $\tilde{x}$  closest to $x$ in $\Su$.
 Choose some $0< \epsilon <\rch( \Su )$ and, as in Lemma  \ref{lemma:SemiContinuitySpecialCase},
take $\alpha>0$ so that \eqref{eq:SemiContinuitySpecialCase} holds and, moreover, such that
\[
\alpha < \rch_{loc.}(\Su) - \rch( \Su ).
\]
Since $x\in  \overline{\ax(\Su)}$, there exists $y \in  \ax(\Su) \cap B^{\circ}( x, \alpha) )$,
and therefore Lemma  \ref{lemma:SemiContinuitySpecialCase} yields that there exists at least two points $p,q \in B^{\circ}(  \tilde{x} , \epsilon )$ such that
$d(y,p) = d(y,q) = d(y, \Su)$.  

Because $d(p,q)  < 2 \epsilon < 2 \rch( \Su )$, by Theorem \ref{theorem:ReachEquivalentMetricDistorsion},
there is a geodesic from $p$ to $q$ of length at most $\pi \rch( \Su )$,
which is less that $2 \pi \rch_{loc.}(\Su)$. 
But then one can apply Corollary \ref{lemma:ShortGeodesicsAreStraights}, which implies that the length of a geodesic from $p$ to $q$ is at most 
\[ 
\ell_{\max} = 2  \rch_{loc.}(\Su) \arcsin \frac{|p-q|}{2  \rch_{loc.}(\Su) }.
\]
However, since projecting a curve outside a given sphere onto the given sphere only decreases its length,
a curve from $p$ to $q$ outside the ball centered at $y$ whose radius
is at most $\rch( \Su ) +   |y-x| <  \rch( \Su ) +  \alpha < \rch_{loc.}(\Su)$ must have  length strictly greater than $\ell_{\max} $,
a contradiction.

We have so far proven that when $\rch( \Su ) < \rch_{loc.}(\Su)$, then the reach
cannot be realized in $\overline{\ax(\Su)} \setminus  \ax(\Su)$.
One must then have $x\in \ax(\Su)$ and there are at least two points $p\neq q$ closest points to $x$ on $\Su$.

If $|p-q|< 2 \rch(\Su)$ then by Theorem \ref{theorem:ReachEquivalentMetricDistorsion}
 one has that 
$d_{\Su}(p,q) \leq  2 \rch( \Su)  \arcsin \frac{|p-q|}{2 \rch( \Su ) }$ which is less than 
$2 \pi \rch_{loc.}(\Su)$. 
Again,  Corollary \ref{lemma:ShortGeodesicsAreStraights} 
implies that the length of a geodesic from $p$ to $q$ is
a most $ 2  \rch_{loc.}(\Su) \arcsin \frac{|p-q|}{2  \rch_{loc.}(\Su) }$,
which, again, is impossible without entering the ball 
centered at $x$ with radius
 $\rch( \Su ) < \rch_{loc.}(\Su)$.

Since $|p-q| \leq 2 \rch(\Su)$, the only possibility is that $|p-q| = 2 \rch(\Su)$,
i.e. $x= (p+q)/2$, a bottleneck.
\end{proof}

\bibliography{geomrefs}

\appendix

\section{Proofs}

%
%
%

\subsection{Proofs}
\label{section:ProofOfLemmaReachOfProjection}

\begin{proof}[Proof of Lemma \ref{lemma:HomologyCCone1}]
If $A$ is a $k$-dimensional vector space, then 
$A \cap  \Sphere^{d-1} = \Sphere^{k-1}$ whose $(k-1)$-homology is non-trivial and $(k'-1)$-homology for $k'\neq k$ is trivial.
Of course, if $k=0$, then $A \cap \Sphere^{d-1} = \emptyset$.

If $A$ is not a vector space, then by Observation \ref{observation:ConvexConeSymetricIFFVectorSpace},
$A \neq -A$. This implies that there exists $a_0 \in A \cap \Sphere^{d-1}$ with $-a_0 \notin A$. This in turn gives that for any $a\in A\cap \Sphere^{d-1}$, 
one has $\angle a_0, a < \pi$ so that the (unique) minimal geodesic from $a$ to $a_0$ remains in $A \cap \Sphere^{d-1}$. 
It follows that there is a deformation retract from $A \cap \Sphere^{d-1}$ to $a_0$ along geodesics. Hence that $A$ is contractible and has therefore trivial homology.
\end{proof}

\begin{proof}[Proof of Lemma \ref{lemma:HomologyCCone2}] 
Since $\{0\}^\perp = \R^d$, the case $k=0$ is trivial. We assume now $k>0$.
Since $0\in A^\perp$,  $B(0,1) \setminus  A^\perp $ can be written as
\[ 
\bigcup_{u \in \Sphere^{d-1} \setminus  A^\perp } \{ \lambda u  \mid \lambda \in (0,1]\},
\]  where $\Sphere^{d-1}$ denotes the unit sphere centred at the origin. 
This means that there is a deformation retract from $B(0,1) \setminus  A^\perp $ to $\Sphere^{d-1} \setminus  A^\perp$.

Now we use the same argument as in Lemma \ref{lemma:HomologyCCone1}. If $k\geq 1$ and $A$ is a $k$-dimensional vector space, then $A^\perp$ is a $d-k$ dimensional vector space and $\Sphere^{d-1} \setminus  A^\perp$ is homeomorphic to the open cylinder $\Sphere^{k-1} \times (-1,1)^{d-k}$, whose $(k-1)$-homology is non-trivial. 

Conversely, $k\geq 1$ and $A$ is not a $k$-dimensional vector space, then by Observation \ref{Obs:dualVspace} neither is $A^\perp$. As in the proof of Lemma \ref{lemma:HomologyCCone1} we note that by Observation \ref{observation:ConvexConeSymetricIFFVectorSpace},
$A ^\perp \neq -A^\perp$. 
This implies that there exists $a_1 \in A^\perp \cap \Sphere^{d-1}$ with $-a_1 \notin A^\perp$. This in turn gives that for any $\tilde{a} \in A ^\perp \cap \Sphere^{d-1}$, 
one has $\angle a_1, \tilde{a}  < \pi$ so that the (unique) minimal geodesic from $\tilde{a}$ to $a_1$ remains in $A^\perp  \cap \Sphere^{d-1}$. 
It follows that there is a deformation retract from $A^\perp  \cap \Sphere^{d-1}$ to $a_1$ along geodesics as well as a homotopy from $\Sphere^{d-1} \setminus A^\perp$ to $\Sphere^{d-1} \setminus a_1 \simeq B^{d-1}$, where $B^{d-1}$ denotes the open $d-1$-dimensional ball. Hence $\Sphere^{d-1} \setminus A^\perp$ is contractible and its homology is therefore trivial.

Finally, we consider the case where $k=0$. If $A$ is a $0$-dimensional vector space, that is a point,  then $A^\perp$ is the entire Euclidean space, and therefore $B(0,1) \setminus A^\perp$ is empty.  Conversely if $B(0,1) \setminus A^\perp$ is empty then $A^\perp$ is the entire Euclidean space and its dual is therefore a point. 
\end{proof} 

\begin{proof}[Proof of Lemma \ref{lemma:DefRetractOfBallMinusNormal}]
Since  $\rch \left( \Su \cap B(p, \rho) \right) \geq \rch(\Su)$, we know that $\pi_{\Su \cap B(p, \rho)}$ is continuous (Lemma \ref{Fed4.8.4} or \ref{Fed4.8.8}).
The homotopy $(t,x) \mapsto (1-t) x + t \pi_{\Su \cap B(p, \rho)}(x)$ is then a deformation retract from $B(p, \rho)$ to $\Su \cap B(p, \rho)$. 
Therefore, the only thing we have to prove is that, for $y \in B(p, \rho)$:
\[
y \in p + \Nor(p, \Su)  \iff \pi_{\Su \cap B(p, \rho)} (y) = p.
\]
We note that thanks to Definition \ref{def:4.3and4.4Fed}, we have $\Nor(p, \Su)= \Nor(p, \Su \cap B(p, \rho))$.
The result now follows by Lemma \ref{Fed4.8.12} and the fact that $\rho < \rch(\Su) \leq \rch \left( \Su \cap B(p, \rho) \right)$.   
%
%
%
%
%
%
\end{proof}

\begin{proof}[Proof of Lemma \ref{lemma:NorIsSemiContinuous}]
Take $\rho$ such that $0< \rho < \rch(\Su)$ and denote by $\Su^{\partial \oplus \rho}$ 
 the boundary of the $\rho$-offset of $\Su$, in other words,
\[
\Su^{\partial \oplus \rho} \defunder{=}  \{x \in \R^d \mid \: d(x, \Su) = \rho \}.
\]
For $p \in \Su$, we write $N(p, \Su, \rho)$ for the set of points of $\Su^{\partial \oplus \rho}$ whose projection is $p$, that is
\[
N(p, \Su, \rho) \defunder{=} \{  q \in \Su^{\partial \oplus \rho}\mid \:  \pi_{\Su}(q) = p \}. 
\]
Lemma \ref{Fed4.8.12} yields that if $p \in \Su$, then
\begin{equation}\label{equation:NormalConeFromPointsAtDistanceRho}
\Nor(p, \Su) = \{ \lambda (q - p) \mid \: \lambda \geq 0, \: q \in N(p, \Su, \rho)  \}.
\end{equation}
For $\eta >0$ we denote by $N(p, \Su, \rho)^\eta$  the open offset of $N(p, \Su, \rho)$ in $\Su^{\partial \oplus \rho}$, that is 
\[
N(p, \Su, \rho)^\eta \defunder{=} \Big\{ q \in \Su^{\partial \oplus \rho}, d\big(q, N(p, \Su, \rho) \big) < \eta \Big\}. 
\]
For any integer $j \geq 1$,  the set
\[
\Su^{\partial \oplus \rho}  \cap \pi_{\Su}^{-1} \Big( B(p, 1/j) \cap \Su \Big)
\]
is compact, being bounded and the inverse image of a closed set by the continuous map $\pi_{\Su}$.
Because $\Su^{\partial \oplus \rho}  \setminus N(p, \Su, \rho)^\eta$ is a closed set, the set
\begin{align} 
K_j &\defunder{=} \Su^{\partial \oplus \rho}  \cap \pi_{\Su}^{-1} \Big( B(p, 1/j) \cap \Su \Big) \cap \big( \Su^{\partial \oplus \rho}  \setminus N(p, \Su, \rho)^\eta \big)
\nonumber 
\\ 
&=  \big( \Su^{\partial \oplus \rho}  \setminus N(p, \Su, \rho)^\eta \big)  \cap \pi_{\Su}^{-1} \Big( B(p, 1/j) \cap \Su \Big) 
\nonumber 
\end{align}
is compact, but may be empty.
We have that 
\[
\bigcap_{j\geq 1} \pi_{\Su}^{-1} \Big( B(p, 1/j) \cap \Su \Big)= \pi_{\Su}^{-1} \big( \{ p \} \big) ,
\]
because if $ a \notin \pi_{\Su}^{-1} \big( \{ p \} \big)$, then $\pi_{\Su} (a)= q \neq p$ and thus $a \notin \bigcap_{j\geq 1} \pi_{\Su}^{-1} \Big( B(p, 1/j) \cap \Su \Big)$ for some $j$.
Hence,   since 
\[
\bigcap_{j\geq 1} K_j  = \Su^{\partial \oplus \rho}  \cap \pi_{\Su}^{-1} \big( \{ p \} \big) \cap \big( \Su^{\partial \oplus \rho}  \setminus N(p, \Su, \rho)^\eta \big) = 
N(p, \Su, \rho) \cap \big( \Su^{\partial \oplus \rho}  \setminus N(p, \Su, \rho)^\eta \big) = \emptyset,
\]
we have that $\big(K_j\big)_{(j \geq 1)}$ is a nested (in inverse order) sequence of compact sets whose intersection is empty. Therefore, there is some 
integer number $n \geq 1$ such that $K_n = \emptyset$, such that
\[
\Su^{\partial \oplus \rho}  \cap \pi_{\Su}^{-1} \Big( B(p, 1/n) \cap \Su \Big) \subset  N(p, \Su, \rho)^\eta.
\]
We note that such an $n$ exists for any $\eta >0$. This implies that for any $\eta$ such that $0< \eta <  \rho \sin  \frac{\epsilon}{2}$, there exists an $n$ large enough such that both:
\begin{equation} \label{equation:NormalConeConditionOnP}
\frac{1}{n} < \frac{1}{2}  \rho \sin \epsilon
\end{equation}
and
\begin{equation} 
\Su^{\partial \oplus \rho}  \cap \pi_{\Su}^{-1} \big( B(p, 1/n) \cap \Su \big) \subset  N(p, \Su, \rho)^{\rho \sin  \frac{\epsilon}{2} }. 
\nonumber
\end{equation} 
Using this $n$ we pick $\alpha = \frac{1}{n}$, and we claim that this $\alpha$ satisfies the assertion in the lemma.
Indeed, if $p' \in \Su \cap B(p, \alpha)$, then $N(p', \Su, \rho) \subset N(p, \Su, \rho)^{\rho \sin  \frac{\epsilon}{2} }$. Hence by definition,
 if $q' \in N(p', \Su, \rho)$, then there exists a
 $q \in N(p, \Su, \rho)$ such that $d(q, q') < \rho \sin  \frac{\epsilon}{2} $,  while by \eqref{equation:NormalConeConditionOnP}, 
$d(p,p') < \rho \sin  \frac{\epsilon}{2} $.
Since $| q-p |= | q' - p'| = \rho$ and $|p-p'|, |q-q'| < \rho \sin  \frac{\epsilon}{2}$, we get
\[
| (q-p) - (q'-p')| < 2 \rho \sin  \frac{\epsilon}{2},
\]
so that
\[
\angle (q'-p', q - p) < \epsilon. 
\]
This holds for any    $q' \in N(p', \Su, \rho)$ and since  $q \in N(p, \Su, \rho)$,
the claim is proved, thanks to \eqref{equation:NormalConeFromPointsAtDistanceRho}.
\end{proof}

\begin{proof}[Proof of Lemma \ref{lemma:anglebetweenKVectorSpacesIsSymmetric}]
If $A=B$ the statement is trivially true.
If $A\neq B$, the bisector of $A$ and $B$ contains at least one hyperplane $P$ and $B$ is the image of $A$ by the isometric mirror symmetry with respect to $P$. 
\end{proof}

\begin{proof}[Proof of Lemma \ref{lemma:embeddedManifoldPositiveReachThenC1Embedded}]
Lemma \ref{lemma:embeddedManifoldPositiveReachThenTangentSpaceIsVectorial}
says that $\Tan(q, \M)$ is a $n$-dimensional vector space for any $q \in \M$.

Regarding $\Tan(q, \M)$ as an element of the Grassmannian $\Grass(n, \R^d)$, 
Lemmas \ref{lemma:NorIsSemiContinuous} and \ref{lemma:anglebetweenKVectorSpacesIsSymmetric} give
 that the map defined from $\M$ to  $\Grass(n, \R^d)$ by  $q \mapsto \Tan(q, \M)$
is continuous. 
It follows that, for $p\in \M$, there is some $\epsilon$, with $0< \epsilon < \rch(\M) /2$, such that
 \begin{equation}\label{equation:SmallAngleBetweenTangentSapces}
\forall q \in \M \cap B(p, \epsilon), \angle \Tan(q, \M) , \Tan(p,\M) < \pi/4.
 \end{equation}
We  claim that for this $\epsilon$, the restriction to $\M \cap B(p, \epsilon)$ of the orthogonal projection on $p+\Tan(p, \M)$ is injective.

Indeed, assume (to derive a contradiction) that there exist $q_1,q_2 \in \M \cap B(p, \epsilon)$ 
such that $q_1 \neq q_2$ and $\pi_{p+\Tan(p, \M)} (q_1) = \pi_{p+\Tan(p, \M)} (q_2)$.
Then, because $|q_2-q_1| < \rch(\M)$, 
inequality \eqref{AngleSegmentToT} from Lemma \ref{Lem:distanceToTangent} yields that
$\sin \angle  (q_1- q_2, \Tan(q_1, \M) ) \leq \frac{1}{2}$. 
Combining this with \eqref{equation:SmallAngleBetweenTangentSapces},
\begin{align*}
\angle  (q_1- q_2, \Tan(p, \M) ) &\leq \angle  (\Tan(p, \M), \Tan(q_1, \M) ) +  \angle  (q_1- q_2, \Tan(q_1, \M) ) \\
&<  \frac{\pi}{4} + \frac{\pi}{6} < \frac{\pi}{2}, 
\end{align*}
which cannot be since $q_2 - q_1$ is orthogonal to $ \Tan(p, \M)$. Hence, the claim is proved.

Since the restriction to $\M \cap B(p, \epsilon)$ of the orthogonal projection on $p+\Tan(p, \M)$ is an  injective continuous  map, from an
open $n$-manifold to the $n$-manifold  $p+\Tan(p, \M)$, invariance of domain implies that it is an homeomorphism on its image.
This means that $\M \cap B(p, \epsilon)$ is the graph of a map $x \mapsto \phi(x)$ from some open set $U_p \subset \Tan(p, \M)$ to $\Nor(p,\M)$, i.e. $x \mapsto p + x + \phi(x)$ is a local homeomorphism from $U_p$ to $\M$. Moreover, these maps induce an atlas for $\M$. 

Definition \ref{def:4.3and4.4Fed} and the fact that $\Tan(p,\M)$ is a $n$-dimensional linear space together imply that $\phi$ is differentiable at $0$ with $D\Phi (0) = 0$, 
 and more generally that $\phi$ is differentiable at any $x \in U_p$. The continuity \eqref{equation:SemiContinuityTangent} of $q \mapsto \Tan(q, \M)$ for the Grassmannians metric yields that $\phi$ is of class $C^1$.
\end{proof}

\begin{proof}[Proof of Lemma \ref{lemma:ProjectionOfTangentConeBoundaryEqualProjOnTgtSpace}]
Thanks to Lemma \ref{lemma:SmallLemma} we have that for sufficiently $\alpha'$ and $q \in  \M \cap \partial B(p, \alpha')$, the bound  $\angle q-p, \Tan(q, \M) < \pi/2$ holds. 
Since $\Nor(q, B(p, \alpha')) = \{ \lambda(q-p)\mid \lambda\geq 0 \}$, this means in particular that 
one has:\
\[
\Nor(q, B(p, \alpha')) \cap \big(-\Nor(q, \M)\big) = \Nor(q, B(p, \alpha')) \cap \Nor(q, \M) = \{0\}.
\] 
This means that one can apply
 \cite[Theorem 4.10 (3)]{Federer}, which gives us
 \begin{align*}
 \Nor \big(q, \M \cap B(p, \alpha' ) \big) &= \Nor \big(q, \M \big) \oplus  \Nor \big(q, B(p, \alpha' ) \big) \\
 &=  \Nor \big(q, \M \big) \oplus  \{ \lambda(q-p)\mid \lambda\geq 0 \}.
  \end{align*}
Because normal and tangent cones are dual we see that
\begin{align}
 \Tan \big(q, \M \cap B(p, \alpha' ) \big) &= \left( \Nor \big(q, \M \big) \oplus  \{ \lambda(q-p)\mid \lambda\geq 0 \} \right)^\perp 
\nonumber
\\
 &=   \Nor(q, \M)^\perp  \cap  \{ \lambda(q-p)\mid \lambda\geq 0 \}^\perp 
\nonumber
\\
 &=  \Tan(q,\M) \cap  \{ u \in \R^d \mid \langle u, q-p \rangle  \leq 0 \}. 
\label{RewriteTanConeIntersectionBall}
\end{align}


\begin{figure}[!h]
\begin{center}
\includegraphics[width=0.95\textwidth]{
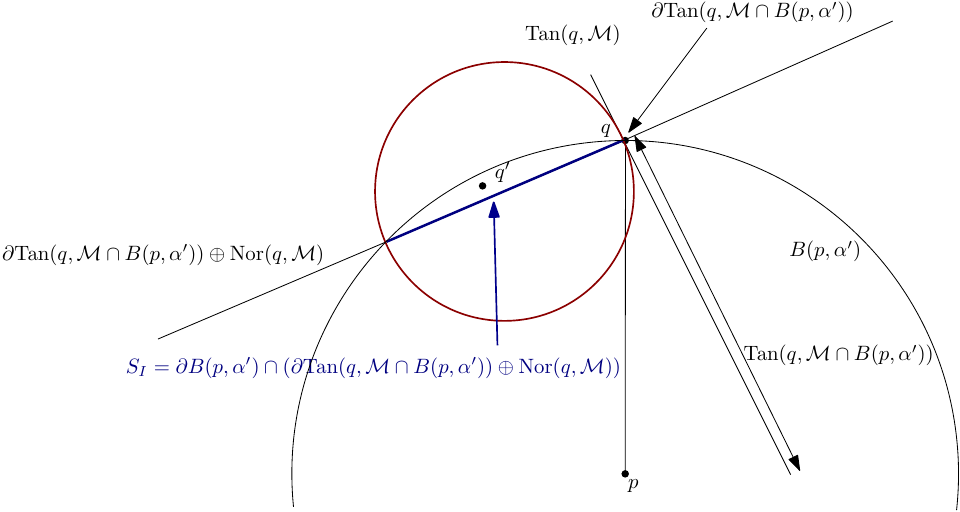}
\end{center}
\caption{Notation used in the proof of Lemma \ref{lemma:ProjectionOfTangentConeBoundaryEqualProjOnTgtSpace} (side view).  
} 
\label{fig:NotationProofL23A}
\end{figure}

\begin{figure}[!h]
\begin{center}
\includegraphics[width=0.9\textwidth]{
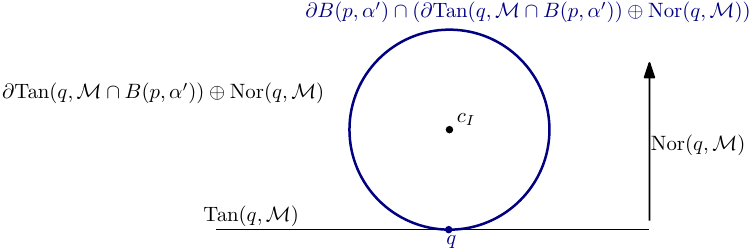}
\end{center}
\caption{Notation used in the proof of Lemma \ref{lemma:ProjectionOfTangentConeBoundaryEqualProjOnTgtSpace} (the restriction to the $d-1$-plane $\partial  \Tan \big(q, \M \cap B(p, \alpha' ) \big) \oplus \Nor(q,\M)$). 
} 
\label{fig:NotationProofL23B}
\end{figure}

Consider the $d-1$ dimensional plane 
\begin{align}  
P&= \partial  \Tan \big(q, \M \cap B(p, \alpha' ) \big) \oplus \Nor(q,\M).
\nonumber
\\
&= \partial  \left( \Tan(q,\M) \cap  \{ u \in \R^d \mid \langle u, q-p \rangle  \leq 0 \}\right) \oplus \Nor(q,\M)
\tag{by \eqref{RewriteTanConeIntersectionBall}}
\\
&=  \left( \Tan(q,\M) \cap  \{ u \in \R^d \mid \langle u, q-p \rangle =0 \}\right) \oplus \Nor(q,\M),
\nonumber
\end{align} 
see Figure \ref{fig:NotationProofL23A}.
We stress that $ \{ u \in \R^d \mid \langle u, q-p \rangle =0 \}$ is the tangent plane of the sphere $\partial B(p, \alpha' )$ at $q$. Moreover the plane $P$ separates exactly those points $q'$ for which the closest point projection of $q'$ on $q \, + \Tan(q, \M)$ belongs to ${ q \, +} \Tan(q, \M \cap B(p, \alpha' ) )$ from the points $q'$ for which this does not hold.   
The plane $P$ intersects the sphere $\partial B(p, \alpha' )$ in $q$, more precisely, the intersection $S_I=q+ P \cap \partial B(p, \alpha' )$ is a $d-2$-sphere that is tangent to $q$. 
These two observations mean that the centre $c_I$ of $S_I$  satisfies $c_I-q \in \Nor (q,\M)$, see Figure \ref{fig:NotationProofL23B}. 
On the other hand any point $q'$ that does not project onto a point that belongs in $q  + \Tan(q, \M \cap B(p, \alpha' ) )$ will lie in the centred at $c_I$, now seen as a point in $\mathbb{R}^d$, with radius $|c_I -q|$, that is the red sphere in Figure \ref{fig:NotationProofL23A}. However, by assumption $\alpha <\rch(\M)$, which means that such a ball cannot contain points of $\M$. 

\end{proof}

\begin{proof}[Proof of Lemma \ref{lemma:VDotphiIsSemiConvex}] 
Let us choose $\alpha>0$ as in Lemma \ref{lemma:SmallLemma} and 
\ref{lemma:ProjectionOfTangentConeBoundaryEqualProjOnTgtSpace},
 so that \eqref{eq:BoundYTheta}  and \eqref{eq:claim_intersect_tangent_to_manifold} 
 hold.
From the definition of $\rch_{loc.}(p, \M)$, $\alpha>0$ can moreover be chosen small enough such that
\begin{align} 
\rch( \M \cap B(p, \alpha)) > R := \rch_{loc.}(p, {\M} ) - \varepsilon/2.
\label{eq:DefRFromrchLoc} 
\end{align}

Let $y \in U_p$ and write $q = p + y + \phi(y) = \Phi (y)\in \M$. 
Choose $\theta >0$ such that:
\begin{align} 
 \cos^3 \theta > \frac{R- \epsilon/2}{R}. 
\label{eq:boundTheta} 
\end{align} 
For $y, y' \in U_p$ we denote the images under $\Phi$, that is the points on $\M$, by $q= \Phi(y)$, $q'=\Phi(y')$.
We now take $y, y' \in U_p$, such that one has $q,q'\in B(p,\alpha)$. This means in particular that $|y|,|y'| \leq \alpha$.

Thanks to
Equation \eqref{TangentSpaceLine_2} in Lemma \ref{Lem:distanceToTangentCharacter} and the fact that $R \leq \rch(B(p,\alpha))$, one has 
 \begin{equation*}
 d (q' ,q  + \Tan( q, \M) )  \leq \frac{|q' - q |^2}{2\, R}.
 \end{equation*}
Combining this with \eqref{eq:claim_intersect_tangent_to_manifold} yields
 \begin{equation}\label{equation:ManifoldNearTangentSpace}
d (q' ,q  + \Tan( q, \M \cap B(p,\alpha)) ) = d (q' ,q  + \Tan( q, \M) )  \leq \frac{|q' - q |^2}{2\, R}.
 \end{equation}

The first order Taylor expansion
of $\Phi$ near $y$ is given by
\begin{align} 
z \mapsto \Phi(y) + D\Phi(y) ( z-y) &= (p + y + \phi(y)) + (z-y) + D\phi(y) (z-y) 
\nonumber 
\\ 
&= p  + \phi(y) + z  + D\phi(y) (z-y)  ,
\nonumber
\end{align}  and its graph
 coincides with the affine tangent space $q + \Tan (q, \M)$ of $\M$ at $q = \Phi(y)$.

Write $q_1$ for orthogonal projection of $q'$ on $q  + \Tan( q, \M) $, that is $ q_1 = \pi _{q  + \Tan( q, \M)} (q')$. We note that $q_1 -q' \in \Nor(q, \M)$.
Since, $q'+ \Nor(p, \M)$ and $q+ \Tan(q, \M)$ are transversal and their dimension sum to $d$, they have a unique intersection point
$q_2$.
The point $q_2$,
 see Figure \ref{fig:GraphOfSemiConvex}, is also 
the point above $y'$ on the graph of the first order Taylor expansion 
of $\phi$ near~$y$: 
\[q_2 = p+ y' + \phi(y) + D\phi(y) (y'-y). \]

\begin{figure}[!h]
\begin{center}
\includegraphics[width=0.8\textwidth]{
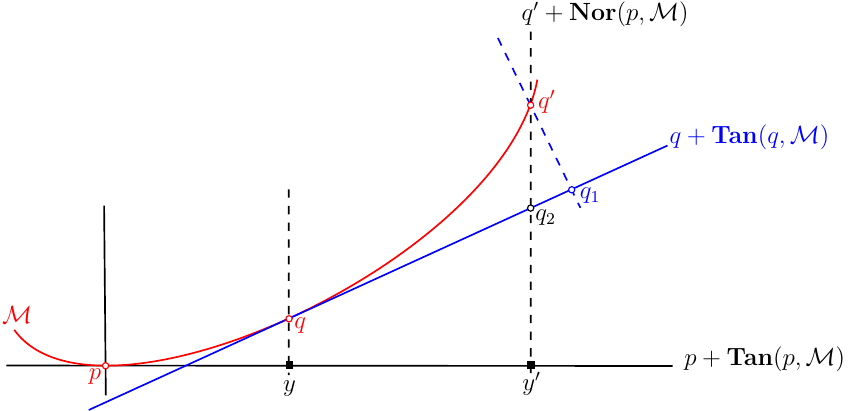}
\end{center}
\caption{Notation used in the proof of Lemma \ref{lemma:VDotphiIsSemiConvex}. 
} 
\label{fig:GraphOfSemiConvex}
\end{figure}

Since $q_1, q_2 \in q + \Tan (q, \M)$, one has that $q_2 - q_1 \in \Tan(q, \M)$. As we concluded above $q_1 -q' \in \Nor (q, \M)$. Combining these observations gives that the triangle $q'q_1q_2$ has a right angle at $q_1$.

Because $q_2-q_1 \in  \Tan(q, \M)$ and $\angle  \Tan(q, \M)  \Tan(p, \M) < \theta$,
 there is a vector $w\in  \Tan(p, \M)$ such that $\angle w, q_2-q_1  < \theta$.
As $q_2-q' \in \Nor(p, \M)$, one has  $\angle w,  q_2-q'  = \pi/2$, and therefore the triangle inequality on the sphere yields $\angle q_2-q', q_2-q_1 > \pi/2 - \theta$. This is one angle of the triangle $q'q_1q_2$ and we already established that this triangle has a right angle at $q_1$ and hence, 
\[
\angle q_2-q', q_1-q'  < \theta. 
\]
This angle bound translates to the {first inequality between the edge lengths of the triangle  $q'q_1q_2$, 
because $ q_1 = \pi _{q  + \Tan( q, \M)} (q')$ and \eqref{equation:ManifoldNearTangentSpace} the second inequality below follows,} 
 \begin{equation}\label{eq:Inequality1}
 |q' - q_2 | \leq \frac{1}{\cos \theta}  |q' - q_1 |  \leq \frac{1}{\cos \theta} \frac{|q' - q |^2}{2\, R}.
 \end{equation}

Since  $\angle q'-q, \Tan(p, \M) < \theta$ by \eqref{eq:BoundYTheta}, and $y'-y$ 
is the projection of $q'-q$ on $p+\Tan(p, \M)$,
we have 
\[
|q' - q | \leq \frac{1}{\cos \theta} |y' - y|.
\]
and \eqref{eq:Inequality1} gives
\[
 |q' - q_2 | \leq \frac{|y' - y |^2}{2\, R \cos^3 \theta}. 
\]
Because $q' = \Phi(y') = p + y' + \phi(y')$,  $|q' - q_2 |$ is the difference between $\phi(y')$ and the value, at $y'$, of the
first order Taylor expansion of  of $\phi$ near $y$, so the last inequality gives
\begin{align*}
 \left|\phi(y') - \Big(\phi(y) + D\phi(y) (y'-y) \Big)\right| &\leq \frac{|y' - y |^2}{2\, R \cos^3 \theta}
\\
& <  \frac{|y' - y |^2}{2\, (R -\epsilon/2)} . \tag{by \eqref{eq:boundTheta}}
\end{align*}
Denoting by $f: U_p \rightarrow \R$ the map defined by $f(y) \defunder{=} \langle v, \phi(y) \rangle$ this together with  $| v| = 1$ gives
\[
 \left| f(y') - \Big(f(y) + Df(y) (y'-y) \Big)\right| <  \frac{|y' - y |^2}{2\, (R -\epsilon/2)} ,
\]
so that
\begin{equation}\label{equation:GraphIsAboveItsTangent}
 \frac{|y' - y |^2}{2\, (R -\epsilon/2)} + f(y') > f(y) + Df(y) (y'-y) .
 \end{equation}
Since both $y' \mapsto |y' - y |^2$ and its derivative vanish at $y$, the right-hand side in \eqref{equation:GraphIsAboveItsTangent}
 is the first order Taylor expansion at $y$
  of the left-hand side $y' \mapsto \frac{|y' - y |^2}{2\, (R -\epsilon/2)} + f(y')$. 
We have shown that the map $y' \mapsto  \frac{|y' - y |^2}{2\, (R -\epsilon/2)} + f(y')$ is greater or equal to  its first order Taylor expansion at $y$.
 
Because $y'\mapsto (y' - y)^2 - y'^2$ is a linear function, this property extends to the map $g: U_p \cap B(0, \alpha)^\circ \rightarrow \R$ defined by
 \[
 g(y') \defunder{=}  \frac{y'^2}{2\, (R -\epsilon/2)} + f(y') =  \frac{y'^2}{2\, (R -\epsilon/2)} + \langle v, \phi(y) \rangle
 =   \frac{y'^2}{2\, ( \rch_{loc.}(p, \Su) -\epsilon)} + \langle v, \phi(y) \rangle,
 \]
{where the last equality follows from \eqref{eq:DefRFromrchLoc}.} 
This holds for any $y \in U_p$ such that  $\Phi(y) \in \cap B(p, \alpha)$, 
so that for some $\alpha'>0$ small enough, 
it holds for any $y \in U_p \cap B(0,\alpha')$. 
There, the function $g$ is $C^1$ and greater or equal to its first order expansion at any point, and is therefore convex.
 We have proven the first statement of the lemma.


 

We now concentrate on the second statement. 
\newline As we have seen in Remark \ref{remark:ReachLocalWellDefined}, $\rho \mapsto  \rch \big( \M \cap B(p, \rho ) \big) $ is non-increasing  
and hence  
 \begin{align} 
 \rch \big( \M \cap B(p,  \eta) \big)  &<  \rch_{loc.}(p, \Su)
\tag{because non-increasing}
\\
&< r.\tag{by assumption}
 \end{align} 

 Reusing the same construction (see Figure \ref{fig:GraphOfSemiConvex}) as for the proof of the first statement of the lemma,
 we see that by \eqref{TangentSpaceLine_2},  there exists $q,q'\in \M \cap B(p, \eta ) $ with $q'\neq q$ such that:  
 \[
r >  \frac{|q-q'|^2}{2 d_{\mathbb{E} } (q' ,\: q \: + \Tan(q, \M \cap B(p, \eta) ) },
\]
 that is
\begin{equation}\label{eq:DistToTgtGreater}
 d_{\mathbb{E} } (q' ,\: q \: + \Tan(q, \M \cap B(p, \eta) ) >  \frac{|q-q'|^2}{2r}.
 \end{equation}

%
%
%

By Lemma \ref{lemma:ProjectionOfTangentConeBoundaryEqualProjOnTgtSpace}, one has
\begin{equation}\label{eq:DistToTgtGreater_2}
 d_{\mathbb{E} } (q' ,\: q \: + \Tan(q, \M ) >  \frac{|q-q'|^2}{2r}.
 \end{equation}
Looking at Figure \ref{fig:GraphOfSemiConvex}, we see that \eqref{eq:DistToTgtGreater_2} translates into
\[
|q'-q_1|  >  \frac{|q-q'|^2}{2r},
\]
 which is equivalent to 
\[ \frac{|q-q'|^2}{2r |q'-q_1|} < 1 .\]
Because the left hand side is strictly less than $1$, there exists a $\theta'>0$ such that $\cos \theta' >  \frac{|q-q'|^2}{2r |q'-q_1|}$. 
This means that in the context of Lemma \ref{lemma:SmallLemma},
one can choose $\eta$ (where we take $\alpha =\eta$ in Lemma \ref{lemma:SmallLemma}) 
small enough to get $\cos \theta >  \frac{|q-q'|^2}{2r |q'-q_1|}$. 
Because the angle between $\Tan(p,\M)$ and $\Tan(q,\M)$ is upper bounded by $\theta$ examining the triangle $q'q_1q_2$ in Figure \ref{fig:GraphOfSemiConvex} yields that $|q'-q_2| \geq |q'-q_1| \cos \theta$, which combined with the previous result yields 
\[
|q'-q_2|  >  \frac{|q-q'|^2}{2r}. 
\]
Writing 
\[
v = \frac{q'-q_2}{|q'-q_2|} \in \Nor(p,\M),
\]
so that $|q'-q_2| = \langle v, q'-q_2\rangle$ and  using 
 $|q-q'| > |y'-y|$, this can be reformulated into 
\begin{align}\label{eq:far_first_order}
 \langle v, q'-q_2\rangle  \geq  \frac{|y'-y|^2}{2r}.
\end{align}

Because
\[
 q' - q_2 = \Big( 0, \: \phi(y') - \big(\phi(y) + D\phi(y) (y'-y) \big) \Big)_{\Tan_p\M \times \Nor_p\M},
 \]
and $f(y) \defunder{=} \langle v, \phi(y) \rangle$,
 we may rewrite \eqref{eq:far_first_order} as
\begin{align}\label{eq:below_first_order}
 \frac{|y' - y |^2}{2r} + f(y') <  f(y) + Df(y) (y'-y) . 
\end{align}

The function  $x \mapsto  g(x) = \frac{|x - y |^2}{2r} + f(x)$
is not convex since \eqref{eq:below_first_order}
says that $g(y')$ is below the first order Taylor expansion of $g$ at $y$.
\end{proof}

\begin{proof}[Proof of Lemma \ref{lemma:phiIsC11}]
We get from Lemmas \ref{lemma:SemiConcaveAndSemiConvexIffC11} and \ref{lemma:VDotphiIsSemiConvex},
and reusing the same $\epsilon$ and $\alpha$ selection as in Lemma \ref{lemma:VDotphiIsSemiConvex},
 we get that for any $v \in \Nor(p, \M)$ such that $| v | = 1$, 
the map $y \mapsto \langle v, \phi(y) \rangle$ is $\frac{1}{(\rch_{loc.}(p, \Su)  - \epsilon)}$-Lipschitz in $U_p \cap B(0, \alpha)^\circ$.

It follows that if $y_1, y_2 \in U_p \cap B(0, \alpha)^\circ$, then, for any $w \in \Tan(p, \M)$ and any $v \in  \Nor(p, \M)$ such that $|w|=|v|=1$:
\[
\langle v, D\phi(y_2) w \rangle - \langle v, D\phi(y_1) w \rangle \leq\frac{1}{(\rch_{loc.}(p, \Su)  - \epsilon)}  | y_2 - y_1 |
\]
and this precisely means that, for the operator norm:
\[
\big\|D\phi(y_2)  - D\phi(y_1)  \big\|  \leq \frac{1}{(\rch_{loc.}(p, \Su)  - \epsilon)}  | y_2 - y_1 |.
\]
\end{proof}

\begin{proof}[Proof of Lemma \ref{lemma:FromOperatorNormToAngles}]
The angle between $L_1$ and $L_2$ is a attained in a pair of vectors, that is there exists $v_1 \in L_1$ and $v_2 \in L_2$ such that
\[
 \angle L_1, L_2 = \angle v_1, v_2 = \min_{v'_2 \in L_2} \angle v_1, v'_2 = \max_{v'_1 \in L_1}  \min_{v'_2 \in L_2} \angle v'_1, v'_2 =\max_{v'_2 \in L_2}  \min_{v'_1 \in L_1} \angle v'_1, v'_2, 
\]
where the last equality follows from Lemma \ref{lemma:anglebetweenKVectorSpacesIsSymmetric}.  This explains the asymmetrical expression. 

For some $v_0 \in \R^n$, with $|v_0| = 1$, and $\lambda >0$, 
we have  $v_1=  \lambda (v_0, F_1(v_0))$ and therefore
\begin{equation}\label{equation:angleLinearSpacesUpperBoundedByAngleImageSameVector}
 \angle L_1, L_2 = \angle (v_0, F_1(v_0)) , v_2 \leq \angle  (v_0, F_1(v_0)) ,  (v_0, F_2(v_0)),
\end{equation}
because, by definition of $F_2$, we have $(v_0, F_2(v_0)) \in L_2$.

Since $| (v_0, F_1(v_0)) |, | (v_0, F_2(v_0)) |  \geq 1$, one has:
\begin{align*}
2 \sin  \frac{  \angle  (v_0, F_1(v_0)) ,  (v_0, F_2(v_0))     }{2} 
&\leq \big|  (v_0, F_2(v_0)) - (v_0, F_1(v_0)) \big| \\
& = \big|    F_2(v_0) - F_1(v_0) \big|  \leq \big\| F_2 - F_1 \big\|.
\end{align*}
This with \eqref{equation:angleLinearSpacesUpperBoundedByAngleImageSameVector}
gives the inequality of the lemma.
\end{proof}

\begin{proof}[Proof of Lemma \ref{lemma:MIsC11WithOptimalBounds}]
{
If $\M$ has local reach $\rch_{loc.}({\M}) =\infty$, \eqref{eq:MetricDefREach} implies that all geodesics are straight lines, that is every connected component of $\M$ is convex. 
Therefore, because $\M$ does not contain a boundary by assumption, every connected component of $\M$ is affine and the statement holds. This means that we can concentrate on the case where $\rch_{loc.}({\M})<\infty$. 
}

Lemma \ref{lemma:phiIsC11} already says that  $\M$ is $C^{1,1}$ embedded by considering the atlas defined by projections on local tangent spaces.
Now, (if $p_1$ and $p_2$ lies in the same connected component of $\M$) considers (thanks to Lemma \ref{Menger}) a geodesic $\gamma$ of length $L= d_{\M} (p_1, p_2)$
joining $p_1$ and $p_2$. Take $\epsilon >0$ and $\theta >0$.  By application of Lemma \ref{lemma:phiIsC11} there is at each point $q\in \gamma([0,L]) $ an open neighborhood $V_q$ of 
$q$ in $\M$, such that $V_q$ is the graph of a map $\phi_q$ from $\Tan(q, \M)$ to $\Nor(q, \M)$ whose derivative is bounded by
 $\tan \theta$ and, if $R\leq \rch_{loc.}(\M)$, is $\frac{1}{(R - \epsilon)}$-Lipschitz (in the sense of Lemma 
\ref{lemma:phiIsC11}). Consider for each $t \in [0,L]$ an open interval $(t^-, t+)$ such that $\gamma((t^-, t+) \cap [0, L] ) \subset V_{\gamma(t)}$. 
By compactness, one can extract a finite open cover of $[0,L]$ by intervals in the form $(t_i^-, t_i^+)_{i=0, N}$, with $t_i \in (t_i^-, t_i^+)$, $t_0 = 0$ and $t_N = L$.
One can therefore find a finite sequence $0=t'_0< \ldots <t'_{N'} = L$ such that for any $j$ such that $0 \leq j < N'$, there is 
some $i$ such that $[t'_j, t'_{j+1}] \subset (t_i^-, t_i^+)$.
Since $\gamma(t'_0) = \gamma(0) = p_1$ and $\gamma(t'_{N'}) = \gamma(L) = p_2$, one has:
\begin{equation}\label{equation:angleBoundedByDeviationAlongGeodesic}
\angle \Tan(p_1, \M) , \Tan(p_2,\M) \leq   \sum_{j=0}^{N'-1} \angle \Tan(\gamma (t'_{j}) , \M) , \Tan( \gamma (t'_{j+1}),\M)
\end{equation}
The segment of geodesic $\gamma([t'_j, t'_{j+1}]$ belongs to the graph of the map from $\Tan(\gamma(t_i), \M)$  which
derivative is $\frac{1}{(R - \epsilon)}$-Lipschitz for the operator norm, we get 
from Lemma \ref{lemma:FromOperatorNormToAngles}   that: 
\[
2 \sin \frac{\angle \Tan(\gamma (t'_{j}) , \M) , \Tan( \gamma (t'_{j+1}),\M)  }{2}
\leq  \frac{1}{(R - \epsilon)}   \left| \pi_{\Tan(\gamma(t_i), \M)} (\gamma (t'_{j+1})) - \pi_{\Tan(\gamma(t_i), \M)} (\gamma (t'_{j})) \right| 
\]
and, since:
\[
\left| \pi_{\Tan(\gamma(t_i), \M)} (\gamma (t'_{j+1})) - \pi_{\Tan(\gamma(t_i), \M)} (\gamma (t'_{j})) \right| \leq \left| \gamma (t'_{j+1}) - \gamma (t'_{j}) \right| 
\]
we get:
\[
2 \sin \frac{\angle \Tan(\gamma (t'_{j}) , \M) , \Tan( \gamma (t'_{j+1}),\M)  }{2}
\leq  \frac{1}{(R - \epsilon)} \left| \gamma (t'_{j+1}) - \gamma (t'_{j}) \right| 
\]
Since $\angle \Tan(\gamma (t'_{j}) , \M) , \Tan( \gamma (t'_{j+1}),\M) < 2 \theta$ and $\theta \mapsto \frac{\sin \theta}{\theta} $ is decreasing for positive small $\theta$,  one has:
\[
 2 \sin \frac{\angle \Tan(\gamma (t'_{j}) , \M) , \Tan( \gamma (t'_{j+1}),\M)  }{2} > \frac{\sin \theta}{\theta}  \angle \Tan(\gamma (t'_{j}) , \M) , \Tan( \gamma (t'_{j+1}),\M) 
 \]
 so that \eqref{equation:angleBoundedByDeviationAlongGeodesic} gives:
\[
\angle \Tan(p_1, \M) , \Tan(p_2,\M) \leq   \sum_{j=0}^{N'-1}  \frac{\theta}{(R - \epsilon) \sin \theta}  \left| \gamma (t'_{j+1}) - \gamma (t'_{j}) \right|  \leq   \frac{\theta}{(R - \epsilon) \sin \theta}  d_{\M} (p_1, p_2)
\]
Since the last inequality holds for any $\theta >0$ and $\epsilon >0$, the lemma is proved.
\end{proof}

 \begin{proof}[Proof of Theorem \ref{ReachImpliesQuantifiedLipschitzDerivative}] 
 One direction is given by Lemma \ref{lemma:MIsC11WithOptimalBounds}.
{For the other direction, we distinguish two cases, $\rch_{loc.}(\M)= \infty$ and $\rch_{loc.}(\M)< \infty$. 
If 
\[ \sup_{ \myatop{p,q\in \M}{p\neq q} } \frac{\angle \Tan(p,\M),  \Tan(q,\M)}{d_{\M}(p,q)} = 0 ,\] 
then $\M$ is an affine space, because it is a closed manifold without boundary. This means that each connected component of $\M$ has infinite reach and by extension the local reach is also infinite. 
This means that we can concentrate on the case where $\rch_{loc.}(\M)< \infty$.
}
It now suffices  to prove that, if $R> \rch_{loc.}(\M)$, then there exists distinct $q_1,q_2 \in \M$ (i.e. $q_1\neq q_2$) 
 such that
\begin{align}\label{eq:RequiredInequalityForOptimalLipschitz}
\frac{\angle \Tan(q_1,\M),  \Tan(q_2,\M)}{d_{\M}(q_1,q_2)} > \frac{1}{R}.
\end{align}
Let $\theta>0$ be such that
\[
 \cos^3 \theta> \frac{\rch_{loc.}(\M)}{R}.
\]
From the definition of $\rch_{loc.}(\M)$ it follows that there exists $p\in \M$ such that
\begin{align}\label{eq:BoundOnThetaForOptimalLipschitz}
\rch_{loc.}(\M) \leq \rch_{loc.}(p,\M) < r :=    R \cos^3 \theta.
\end{align}

Lemma \ref{lemma:SmallLemma} implies that there exists $\alpha>0$ such that \eqref{eq:BoundYTheta}
holds in the neighbourhood of $p$. By taking $\eta=\alpha$ 
in the second statement of Lemma \ref{lemma:VDotphiIsSemiConvex} we have 
 that the restriction of 
\[
 x \mapsto f(x) := 1/2 \: x^2 -  r  \langle v, \phi(x) \rangle
\]
to $B(0,\eta)$ is not convex. 
{We stress that as in Lemma \ref{lemma:embeddedManifoldPositiveReachThenC1Embedded} (and by extension Lemma \ref{lemma:VDotphiIsSemiConvex}) $p$ has coordinate $0$.}
This implies that there exists $y,y' \in B(0,\eta)$ such that the segment $[(y, f(y)), (y', f(y'))]$
is not always above the graph of $f$. Let $t_1, t_2\in (0,1)$, with $t_1<t_2$  be such that the open interval $(t_1,t_2)$
is a connected component of the non empty open set 
$\{t\in [0,1]\mid f((1-t)y + t y') > (1-t)f(y) + t f(y') \}$.
This means that if we define $y_i= (1-t_i)y + t_i y'$ we have that
\begin{align*}
\frac{d}{dt}_{\mid t=0}  f((1-t)y_1 + t y_2) &\geq  \frac{f(y_2) -f(y_1)}{t_2-t_1} 
\\
\frac{d}{dt}_{\mid t=1}  f((1-t)y_1 + t y_2) &\leq \frac{f(y_2) -f(y_1)}{t_2-t_1}.
\end{align*}
Combining these equations yields
\[
\frac{d}{dt}_{\mid t=1}  f((1-t)y_1 + t y_2) -  \frac{d}{dt}_{\mid t=0}   f((1-t)y_1 + t y_2) \leq 0.
\]
Using the definition of $f$ given above, this translates into 
 \begin{align}\label{eq:DerivativeIsNotIncreasingAlongY1Y2}
(y_2 - y_1)^2 - r  \left( \langle v, \big( D\phi(y_2) - D\phi(y_1) \big)(y_2 - y_1) \rangle \right) \leq 0.
\end{align}

We define $q_i = \Phi(y_i) = (y_i; \phi(y_i))$, for $i=1,2$ and the vectors $v_i := D\Phi(y_i) (y_2 - y_1) \in \Tan_{q_i}\M$.
Note that $\angle \Tan_p\M,  \Tan_{q_1}\M < \theta$.
Further writing $v'_2$ for the projection of $v_2$ on
$\Tan_{q_1}\M$, we consider the triangle $v_2, v'_2, v_1$.

Since $v'_2-v_2$  is orthogonal to  $\Tan_{q_1}\M$, and $v'_2-v_1 \subset \Tan_{q_1}\M$, the triangle $v_2, v'_2, v_1$ is orthogonal at $v'_2$ 
 and  $\angle v'_2-v_2, v_1-v_2 < \theta$. One has then
\[
 |v_2|\sin \angle v_2,  \Tan_{q_1}\M = | v_2- v'_2| \geq \cos \theta  | v_2- v_1|. 
\]
Moreover, because  
\[
\angle \Tan_{q_2}\M,  \Tan_{q_1}\M \geq \angle v_2,  \Tan_{q_1}\M \geq \sin \angle v_2,  \Tan_{q_1}\M ,
\]
and $ |v_2| <  \frac{1}{\cos \theta} |y_2-y_1|$, one has
\[
\angle \Tan_{q_2}\M,  \Tan_{q_1}\M  > \cos^2 \theta \frac{ | v_2- v_1|}{ |y_2-y_1|}.
\]
By \eqref{eq:DerivativeIsNotIncreasingAlongY1Y2} one has
\[
 | v_2- v_1| \geq \langle v,  v_2- v_1 \rangle \geq \frac{1}{r}  (y_2-y_1)^2.
\]
Combining the last two equations yields 
\[
 \angle \Tan_{q_2}\M,  \Tan_{q_1\M}  > \frac{\cos^2 \theta }{ r }  |y_2-y_1|  \geq \frac{\cos^3 \theta }{ r }  d_{\M}(q_1,q_2),
\]
where the last inequality follows from the fact that $d_{\M}(q_1,q_2)$
is upper bounded by the length of $\Phi([y_1,y_2])$.
This with \eqref{eq:BoundOnThetaForOptimalLipschitz} gives us 
\eqref{eq:RequiredInequalityForOptimalLipschitz}.
\end{proof}

\begin{proof}[Proof of Lemma \ref{lemma:LocalReachAsTangentVaraiationBoundOnGeodesics}]
Applying Theorem \ref{theorem:ReachEquivalentMetricDistorsion} to the set $\M\cap B(p,\rho)$,
for arbitrary small $\rho$, shows that the local reach of $\M$ is the infimum on all
the local reach on all geodesics on $\M$.
Then one can apply Lemma \ref{lemma:MIsC11WithOptimalBounds}
to $\gamma([0,\ell])$ which gives \eqref{eq:TightTgtVariationGeodesics}.
\end{proof}

\begin{proof}[Proof of Theorem \ref{PosReachImpliesOptimalSizeNeighbourAsGraph}] 
Thanks to Theorem \ref{theorem:ReachEquivalentMetricDistorsion} we have that for $q \in \M \cap B^\circ(p, \sqrt{2}\: \rch(\M))$
\[ 
d_{\M}(q,p) \leq 2\,  \rch(\M)  \arcsin \frac{|q-p|}{ 2 \rch(\M)} <  2\,  \rch(\M)  \arcsin \frac{\sqrt{2}\: \rch(\M)}{ 2 \rch(\M)} = \frac{\pi}{2} \rch(\M) . 
\] 
Now using Theorem \ref{ReachImpliesQuantifiedLipschitzDerivative} or Lemma \ref{lemma:MIsC11WithOptimalBounds}, this implies that 
 \[
 \angle \Tan(p, \M), \Tan(q,\M) \leq  \frac{1}{ \rch(\M) }  d_{\M} (p, q)<  \frac{\pi}{2} . 
\]
The submersion theorem now yields that the orthogonal projection $\pi _{T_p \M}$ from $\M \cap B^\circ(p, \sqrt{2}\: \rch(\M))$ onto $\Tan(p,\M)= T_p\M$ is a local diffeomorphism.  

We recall that the covering number of a point in the image of a map is the number of points in the preimage above that point. 
Because the orthogonal projection $\pi _{T_p \M}$ is a local diffeomorphism, the covering number is constant on $ B^\circ(0,\rch(\M)) \subset T_p\M$ except when one crosses the boundary $\partial \big( \overline{ \M \cap B^\circ(p, \sqrt{2}\: \rch(\M))}  \big)$. However, we claim that 
\[ 
\partial \big( \overline{ \M \cap B^\circ(p, \sqrt{2}\: \rch(\M))}  \big) \cap \pi_{T_p\M}^{-1} (B^\circ(0,\rch(\M))) = \emptyset. 
\]
This claim can be established as follows: To derive a contradiction we assume that $q\in \partial \big( \overline{ \M \cap B^\circ(p, \sqrt{2}\: \rch(\M))}  \big)$, and write $q=x+n$, with $x \in B^\circ(0,\rch(\M)) \subset T_p\M$ and $n \in N_p\M$. Then by Pythagoras we have that $n> \rch(\M)$.  This contradicts the characterization of the reach given in Lemma \ref{Lem:distanceToTangentCharacter}. 

We therefore conclude that the covering number is constant on $B^\circ(0,\rch(\M)) \subset T_p\M$. However, thanks to Lemma \ref{lemma:embeddedManifoldPositiveReachThenC1Embedded} the covering number is one at $0$ and in fact in a neighbourhood of zero. Hence the covering number is $1$ on $B^\circ(0,\rch(\M)) \subset T_p\M$. Because $\M$ and the graph is locally $C^{1,1}$ it is globally $C^{1,1}$ and the result follows. 
%
\end{proof} 

\begin{proof}[Proof of Lemma \ref{cor:OtherPaper}] 
In the proof of Lemma \ref{lemma:VDotphiIsSemiConvex} we pick $\theta\in [ 0 ,\pi /2]$ such that 
\begin{align} 
 \cos^3 \theta > \frac{R- \epsilon}{R}. 
\tag{\ref{eq:boundTheta}}
\end{align} 
Given $\theta$ we then choose an $\alpha>0$ so that
\begin{align} 
| y  | \leq \alpha \Rightarrow  \angle \Tan(p, \M), \Tan(q, \M) < \theta.
\tag{\ref{eq:BoundYTheta}}
\end{align} 
Because $\frac{d^3}{d\theta^3 } \cos^3 (\theta)= \sin (\theta) ( 21 \cos^2 (\theta) -6 \sin^2 (\theta) ) \geq 0$ if $\theta \in [0,\pi /4]$ , Taylor's Theorem with integral remainder yields that remainder the $R(\theta)$ in
\[ 
\cos^3 (\theta) = 1- \frac{3}{2} \theta^2  + R(\theta)
\] 
is positive, if $\theta \in [0,\pi /4]$. As $\cos^3(\theta)$ is monotone on $[0,\pi/2]$ and 
\[
\cos \left( \frac{\pi}{4} \right  ) = \frac{1}{2 \sqrt{2}} < 1/2 \leq  \frac{R- \epsilon}{R},
\] 
and hence $\theta \in [0,\pi /4]$, 
Therefore it follows that 
\[ 
\cos^3 (\theta) \leq 1- \frac{3}{2} \theta^2 \leq 1-\theta^2,
\]
which in combination with \eqref{eq:boundTheta} gives that $\theta^2 \leq \frac{\epsilon}{R}$. 
We need to bound $\alpha$ such that $\angle \Tan(p, \M), \Tan(q, \M) < \theta$, but thanks to Lemma \ref{lemma:MIsC11WithOptimalBounds} we know that 
 \[
 \angle \Tan(p, \M), \Tan(q,\M) \leq  \frac{1}{R }  d_{\M} (p,q),
\]
so that we need to impose 
\[ 
\frac{1}{R }  d_{\M} (p, q) \leq \sqrt{\frac{\epsilon}{R}} . 
\] 
Because $q = p + y + \phi(y)$, we have that $|y| \leq  d_{\M} (p, q)$, which results in
\[ 
|y| \leq \sqrt{R \epsilon}.
\] 
\end{proof}

\begin{proof}[Proof of  Lemma \ref{lemma:LocalReachLowerSC} ]
By definition of $\rch_{loc.}(p, \Su)$, 
we have that for any $\epsilon>0$ there is $\rho>0$ such that:
\begin{equation}\label{leq:ReachInLocalBallIConvergesToLocalReach}
\rch(p, \Su \cap B(p, \rho) ) > \rch_{loc.}(p, \Su) - \epsilon.
\end{equation}
One gets
\begin{align*}
\rch_{loc.}(p, \Su) &\geq \inf_{q\in \Su \cap B^{\circ}(p, \rho)} \:  \rch_{loc.}(q, \Su) \\
&= \inf_{q\in \Su \cap B^{\circ}(p, \rho)} \:  \rch_{loc.}(q, \Su \cap B(p, \rho) ) \\
&  \geq  \rch_{loc.}( \Su \cap B(p, \rho) ) \tag{by \eqref{eq:DefLocReach2} applied to the set $\Su \cap B(p, \rho)$}  \\
& \geq  \rch( \Su \cap B(p, \rho) ) \tag{by \eqref{leq:LocalReachLargerThanReach}} \\
&> \rch_{loc.}(p, \Su) - \epsilon .  \tag{by \eqref{leq:ReachInLocalBallIConvergesToLocalReach}}
\end{align*}
Because this inequality holds for any positive $\epsilon$, the equality follows. 
\end{proof}

\begin{proof}[Proof of Lemma \ref{lemma:LocalReachMakeGeodesicsStraight}] 
The length of $\gamma$ in the direction $\dot{\gamma}(\frac{\ell}{2})$ is
\begin{align}
\left \langle q-p , \dot{\gamma} \left( \frac{\ell}{2} \right) \right \rangle &= \int_{0}^{\ell} \langle \dot{\gamma}(s), \dot{\gamma}(\ell/2) \rangle \ud s 
\nonumber
\\
&= \int_{0}^{\ell/2} \langle \dot{\gamma}(s), \dot{\gamma}(\ell/2) \rangle \ud s +\int_{\ell/2}^{\ell} \langle \dot{\gamma}(s), \dot{\gamma}(\ell/2) \rangle \ud s
\nonumber 
\\
&\geq \int_{0}^{\ell/2} \cos \frac{|s-\ell/2| } {R}  \ud s +\int_{\ell/2}^{\ell} \cos \frac{|s-\ell/2| } {R}  \ud s 
\nonumber
\\
&= 2 \, R\sin \left(\frac{ \ell }{2 R} \right).
\nonumber
\end{align}
{where, in the above inequality, we use the assumption
$\ell \leq 2 \pi R$ and the fact that  $\theta \mapsto  \cos \theta$ 
is decreasing on $[0,\pi]$.}

Because $|q-p| \geq \langle q-p , \dot{\gamma}(\frac{\ell}{2}) \rangle$, we see that 
\[ 
|q-p| \geq 2 \, R \sin \left(\frac{ \ell }{2R} \right).
\] 
\end{proof}

\begin{proof}[Proof of Lemma \ref{lemma:SemiContinuitySpecialCase}]
Consider the closed set $\Su' := \Su \setminus B^{\circ}(  \tilde{x} , \epsilon )$.
The minimum distance $d(x, \Su')$ to $x$ is attained at least at some point $z$ in $\Su'$
so that
\[
d(x, \Su') = d(x, z) > d(x, \tilde{x} ) = d(x, \Su).
\]
Taking
\[ \alpha := \frac{ d(x, \Su') - d(x, \Su)}{2},
\]
we get, for any $y$ such that $|y-x|< \alpha$ and $z \in \Su'$ that:
\begin{align*}
d(y, z)  &\geq d(y, \Su') \geq d(x, \Su')  - d(y,x) \\
& =  d(x, \Su) + 2 \alpha   - d(y,x)  =   d(x, \tilde{x} )  + 2 \alpha   - d(y,x) \\
&>  d(x, \tilde{x} ) + d(y,x) \geq  d(y, \tilde{x} ) \geq d(y, \Su).
\end{align*}
We have shown that if $z \in \Su'$, then $d(y, z) > d(y, \Su)$. It follows that 
any closest point to $y$ in $\Su$ is in $\Su \setminus \Su' = \Su \cap B^{\circ}(  \tilde{x} , \epsilon )$
which is \eqref{eq:SemiContinuitySpecialCase}.
\end{proof}


\section{The metric on Grassmannians }
\label{app:Grassmann} 

More precisely it is shown in \cite[Section 34]{akhiezer2013theory} that 
\begin{align} 
d_{G}( A, B) =  \max\{ \sup_ {u \in S_A} d(u,B),\sup_ {u \in S_B} d(u,A) \}  =  \|  P_A - P_B\|,
\nonumber
\end{align} 
where $S_A$ is the unit sphere restricted to $A$, $d(u,A)$ denotes the distance from $u$ to $A$ (as above), and $P_A$ denotes the linear (orthogonal) projection operator on $A$. We write $d_G$ to emphasize its use in the context of Grassmannians, but it is defined for arbitrary Hilbert spaces. 
With these definitions we have $d_{G}( A, B) = \sin (\angle A, B)$. Because $d_G$ is now written in terms of an operator norm the $d_G$ is a metric. 
The fact that the maximal angle is a metric can also be seen as follows. By taking the intersection of an $n$-dimensional linear subspace with the unit sphere in $\mathbb{R}^d$, one finds an $n-1$-dimensional sphere. The Hausdorff distance between two such sphere is precisely the angle between the two linear subspaces. Because we already know that the Hausdorff distance is a metric, the result follows. 

{The fact that it is a metric can also be proven directly. The fact that the angle is zero if the spaces coincide follows immediately from the definition, and the same holds for positively, while symmetry was proven in Lemma \ref{lemma:anglebetweenKVectorSpacesIsSymmetric}. The only thing that remains to verify is the triangle inequality.   }
By definition $\angle A,B$ and $\angle B,C$ are the smallest number such that  
\begin{align} 
\forall a \in A \quad \exists b \in B &&\textrm{such that} && \angle a,b &\leq \angle A,B 
\nonumber
\\
\forall b \in B \quad \exists c \in C &&\textrm{such that} && \angle b,c &\leq \angle B,C
\nonumber 
\end{align} 
so that by combining these 
\begin{align} 
\forall a \in A \quad \exists b \in B,  c \in C &&\textrm{such that} &&  \angle a,b +\angle b,c &\leq \angle A,B+ \angle B,C .
\nonumber
\end{align} 
by the triangle inequality between vectors $\angle a,c \leq \angle a,b +\angle b,c $ so that 
\begin{align} 
\forall a \in A \quad \exists c \in C &&\textrm{such that} &&  \angle a,c &\leq \angle A,B+ \angle B,C .
\nonumber
\end{align}
This gives $\angle A,C \leq  \angle A,B+ \angle B,C $, that is the triangle inequality, by definition.

\end{document}